\documentclass[12pt]{amsart}
\usepackage{amsmath}
\usepackage{amsthm}
\usepackage{graphicx}% Include figure files
\usepackage{dcolumn}
\usepackage{bm}
\usepackage{amsmath}
\usepackage{amsthm}
\usepackage{amssymb}
\usepackage{amscd}
\usepackage{bbold}
\usepackage{enumitem}
\usepackage{subfigure}
\usepackage[utf8]{inputenc}
\usepackage[T1]{fontenc}
\usepackage{mathptmx}

\newtheorem{lem}{Lemma}[section]
\newtheorem{thm}[lem]{Theorem}
\newtheorem{prop}[lem]{Proposition}

\theoremstyle{definition}
\newtheorem{defn}[lem]{Definition}
\newtheorem{example}[lem]{Example}

\newtheorem{rem}[lem]{Remark}

\newcommand{\CC}{{\mathbb C}}

\newcommand{\GG}{{\mathbb G}}
\newcommand{\HH}{{\mathbb H}}

\newcommand{\KK}{{\mathbb K}}

\newcommand{\NN}{{\mathbb N}}

\newcommand{\QQ}{{\mathbb Q}}
\newcommand{\RR}{{\mathbb R}}

\renewcommand{\SS}{{\mathbb S}}

\newcommand{\ZZ}{{\mathbb Z}}

\newcommand{\Bcal}{\mathcal{B}}

\newcommand{\Ecal}{\mathcal{E}}
\newcommand{\Fcal}{\mathcal{F}}
\newcommand{\Gcal}{\mathcal{G}}

\newcommand{\Ical}{\mathcal{I}}

\newcommand{\Lcal}{\mathcal{L}}

\newcommand{\Pcal}{\mathcal{P}}

\newcommand{\Scal}{\mathcal{S}}

\newcommand{\Ucal}{\mathcal{U}}

\def\benm{\begin{enumerate}}            % Begin enumerate command
\def\eenm{\end{enumerate}}              % End enumerate command
\newcommand{\norm}[1]{\left\Vert #1\right\Vert}         % Norm command (1 argument)
\newcommand{\inner}[2]{\left\langle #1, #2 \right\rangle}    % Inner Product (2 arguments), i.e.,
\begin{document}
%\title[short title]{title}
\title{ A generalization of Gleason's frame function for quantum 
measurement  }

\date{\today}
%% This inserts today's date.
%% If no date is required just type
%% \date{}

%%\subjclass[]{}

\author{John J. Benedetto}
\email{jjb@math.umd.edu}
% \urladdr{http://www.math.umd.edu/\textasciitilde jjb}
\author{Paul J. Koprowski}
\author{John S. Nolan}
\thanks{We are grateful to Dr. Rad Balu of the Army Research Labs, Adelphi (ARL) for 
telling us about Busch's
work on Gleason's theorem, and for arranging a post-doctoral position
for Dr. Koprowski at ARL. The first-named 
author gratefully acknowledges the support
of ARO Grant W911NF-17-1-0014 and NSF-DMS Grant
18-14253. The second named author gratefully acknowledges the support of the
Norbert Wiener Center and ARL. The third named author gratefully
acknowledges the support of the Norbert Wiener Center as a Daniel Sweet
Undergraduate Research Fellow. We would all like to thank Professor Robert
Benedetto of Amherst College for several key algebraic insights.
Finally, the first named author had the 
unbelievable privilege of having
both Professors Gleason and Mackey as instructors during the period 1960--1962
for real/functional analysis and for Schwartz' distribution theory, respectively.}
\address{Norbert Wiener Center\\
         Department of Mathematics \\
         University of Maryland \\
         College Park, MD 20742 \\
         USA}

\newcounter{bean}
\begin{abstract}
The goal is to extend Gleason's notion of a frame 
function, which is essential in his fundamental theorem 
in quantum measurement, to a more general 
function acting on 1-tight, so-called, Parseval frames. 
We refer to these functions as Gleason functions for Parseval frames.
The reason for our generalization is that positive operator valued
measures (POVMs) are essentially equivalent to Parseval frames, and that
POVMs arise naturally in quantum measurement theory. We prove
that under the proper assumptions, Gleason functions for Parseval frames are 
quadratic forms, as well as other results analogous to Gleason's original theorem.
Further, we solve an intrinsic problem relating Gleason functions for Parseval frames of
different lengths. We use this solution to weaken the hypotheses
in the finite dimensional version of Busch's theorem, that itself is an analog of Gleason's 
mathematical characterization of quantum states.\end{abstract}

\maketitle

\section{Introduction}
\label{sec:intro}

\subsection{Background}
\label{sec:back}

 %Perhaps surprisingly, there are related problems with implications in
%quantum mechanics, and that we intend to
%pursue, in the case of $p$-adic fields.

Garrett Birkhoff and John von Neumann \cite{BirNeu1936} (1936)
introduced  {\it quantum logic} and the role of {\it lattices} to fathom
``the novelty of the logical notions which quantum theory pre-supposes".

The topics they mentioned for this ``novelty" include:
\begin{enumerate}
\item Heisenberg's uncertainty principle,
\item Principle of non-commutativity of observations.
\end{enumerate}

Their fundamental ideas led to the {\it representation theorem} in quantum logic that,
loosely speaking, allows one to treat quantum measurement outcomes as a lattice 
$L(\HH)$ of subspaces of a separable Hilbert space $\HH$ over the field
${\mathbb K}$, where ${\mathbb K} = {\mathbb R}$ or
${\mathbb K}={\mathbb C}$, see, e.g., \cite{BusGraLah1997}, \cite{pito2005}.
As such, the work of Birkhoff and von Neumann, as well as von Neumann's 
classic \cite{vonn1955} led to the study of measures on the closed 
subspaces of $\HH$ as formulated by Mackey \cite{mack1957}, cf. \cite{mack1963}, \cite{mack1978}.

A {\it measure on the closed subspaces of $\HH$} 
is a function $\mu$, that
  assigns, to every closed subspace of $\HH$, a non-negative number such that if $\{X_i\}$
 is a sequence of mutually orthogonal subspaces having closed linear span $X$, then 
 \[
      \mu(X)=\sum_i \mu(X_i).
 \]
 Let ${\rm dim}(\HH)$ denote the dimension of $\HH$.
In \cite{glea1957} (1957), Gleason proved 
the celebrated result that
{\it if ${\rm dim}(\HH) \geq 3$,
then every such measure $\mu$ can be defined as
\begin{equation}
\label{eq:defmu}
  \mu(X) = {\rm tr} (AP_X),
\end{equation}
where $X \subseteq \HH$ is a closed subspace, $P_X$ is the orthogonal
projection onto $X$, {\rm tr} denotes the trace of the operator, and $A$ is a 
positive semi-definite self-adjoint
trace class operator}, see Remark \ref{rem:born}, Subsection \ref{sec:prelim}, 
and Theorem \ref{thm:glea2},
as well as the beautiful proof of Gleason's theorem by Parthasarathy \cite{part1992},
Chapter 1, Section 8.
Going back to von Neumann, $A$ is also referred to as a {\it density operator}
when ${\rm tr}(A) = 1$,
and is often denoted by $\rho$. See Theorem \ref{thm:busch} for Busch's analog
of \eqref{eq:defmu}, that has a different and meaningful definition of
measure allowing the set of projections $P_X$ to be extended to a 
larger set of operators in a physically meaningful way.

If ${\rm dim}(\HH) < \infty$ and 
$B$ is a linear operator on $\HH$ which has matrix
representation $M_B$, then
$B$ is trace class, and the trace of $B$ is
the sum of the diagonal values of $M_B$.
These notions, as well as those introduced in Subsection \ref{sec:gleagoal},
will be expanded upon in the remaining sections. They are given
here in Section \ref{sec:intro} in bare-bones fashion so that we can 
state the goal of the paper in Subsection \ref{sec:gleagoal}. 

Throughout, $\HH$ denotes a separable Hilbert space over $\KK$.
In the $d$-dimensional case, we shall deal exclusively with the Hilbert space 
$\HH = \KK^d$ over $\KK$, taken with the canonical inner product,
since all $d$-dimensional inner-product spaces over $\KK$
are isometric to the Hilbert space $\HH$; 
and we shall not need further refinements 
such as defining different inner products
on the same space in terms of different matrices.

\begin{example}[Closed subspaces of $\HH$]
\label{ex:subsp}
{\it a.} Let $\HH$ be infinite dimensional. The subspaces $X$ of $\HH$ are not necessarily closed.
For example, let $\HH = L^2[a,b]$ and let $X = C[a,b]$.

{\it b.} On the other hand, every subspace $X$ of $\HH = \KK^d$ is closed. To see this, let
${\rm dim}(X) = m <d$ and let $\norm{x_n - y} \rightarrow 0$, $x_n \in X,\, y \in \HH$.
Assume $y \not\in X$. If $\{u_1, \ldots, u_m \}$ is an orthonormal basis for $X$, then
$w = \sum_{j=1}^m \langle y,u_j \rangle u_j$ is the unique vector in $X$ for which 
$\norm{w-y} = {\rm inf}\{\norm{x-y} : x \in X \}$. Because $y \not\in X$, 
we have  
$\norm{x-y} \geq \norm{w-y} > 0$ for all $x\in X$. This contradicts the hypothesis that
$\norm{x_n - y} \rightarrow 0$.
\end{example}

%%%%%%%%%%%%%%%%%%%%%%%%

\subsection{The role of Gleason's theorem and our goal}
\label{sec:gleagoal}

The theory of frames was initiated
by Duffin and Schaeffer in 1952 \cite{DufSch1952}, 
but frames were actually defined by Paley and Wiener in 1934 \cite{PalWie1934} 
to deal with closed linear span problems.
A frame is a natural
generalization of an ONB.
For detailed introductions to frames, see \cite{bene1992}, \cite{CasKut2012},
 \cite{chri2016}.
We now define a frame in order to formulate our goal, and
shall expand on the theory of frames in Subsection \ref{sec:parsframpovm}.
We denote the standard
inner product associated with the Hilbert space $\HH$ by $\inner{\cdot}{\cdot}.$

\begin{defn}[Frames]
\label{defn:framekd} 
Let $\HH =\KK^d$. 

{\it a.} A sequence $\{x_j\}_{j \in J} \subseteq  \HH$ is a {\it frame} for
the Hilbert space $\HH$ if 
\[
    \exists \,A, B > 0, \; \text{such that} \: \forall y\in \HH, \;
   A\norm{y}^2\leq \sum_{j \in J} |\langle y,x_j \rangle|^2\leq B\norm{y}^2.
\]
If $A=B$, then $\{x_j\}_{j \in J}$ is an $A$-{\it tight frame} for $\HH$.
If $A=B=1$, then $\{x_j\}$ is a $1$-{\it tight} or {\it Parseval frame} for $\HH$.
In this case, each $\norm{x_j} \leq 1$, see Proposition \ref{prop:parsprops}.

 The cardinality of the sequence $J$ is 
 denoted by ${\rm card}(J)$, and it satisfies $d \leq {\rm card}(J) \leq \infty$. 
 Usually, our Parseval frames will satisfy $N :={\rm card}(J) < \infty$, but we shall
 need the case ${\rm card}(J) = \infty$ in Theorem \ref{thm:genprob}.

{\it b.} If a sequence $\{x_j\}_{j=1}^d$ is an orthonormal basis (ONB) for the Hilbert space
$\HH$, then Parseval's identity ensures that $\{x_j\}_{j=1}^d$ is a 
Parseval frame for $\HH$, e.g., \cite{GohGol1981}, page 27. Hence, any ONB is a Parseval frame
and we may view Parseval frames as a natural generalization of 
ONBs.

\end{defn}

Gleason's classification
of measures on closed subspaces of Hilbert spaces,
stated in Subsection \ref{sec:back}, depends on his notion of a frame function.
Since this is not related to the theory of frames, we shall refer to such functions as Gleason functions.

\begin{defn}[Gleason function for ONBs]
\label{def:glea}
A {\it Gleason function of weight} $W \in \KK$ {\it for the}
ONBs {\it for} $\HH$ is a function $g : S \longrightarrow \KK$, where $S \subseteq \HH$
is the unit sphere, 
\[
   S := \{ x \in \HH : \norm{x} =1\},
\]
and such that, for all ONBs $\{x_j\}_{j\in J}$ for $\HH$, one has
\[
   \sum_{j \in J}g(x_j)=W.
\]
In the case that $\HH = \KK^d$, the unit sphere $S$
is denoted by $S^{d-1}$.

\end{defn}

\begin{rem}[Quantum logic and the Born model]
\label{rem:born}
In quantum measurement theory, Gleason's theorem has ramifications with regard to the transition
from the quantum logic lattice interpretation of quantum events to a validation of the
{\it Born model} (or {\it rule} or {\it postulate}) for probability in quantum mechanics. 
Specifically, Mackey had asked whether every measure on the lattice of projections
of a Hilbert space can be defined by a positive operator with unit trace.
Kadison proved this is false for $2$-dimensional Hilbert spaces.
Gleason's theorem, and, in particular, (\ref{eq:defmu}), answers Mackey's question
in the positive for higher dimensional Hilbert spaces. This means that
a Gleason function for the ONBs for $\HH = \KK^d$ that is defined by a self-adjoint operator
as in Theorem \ref{thm:saimpliesglea} is compatible with the Born rule,
see,
e.g., \cite{pito2005}, \cite{fuch2011}, \cite{LomForHolLop2017}. In functional analysis, Gleason's
theorem has had significant generalizations with regard to von Neumann algebras and other
abstract notions, see \cite{hamh2003}. These directions are not part of our goal.

\end{rem}

\begin{defn}[Gleason function for Parseval frames]
Let $\HH = \KK^d$. A {\it Gleason function of weight} $W \in \KK$ {\it for the Parseval frames for} $\HH$
 is a function 
 $g:B^d \longrightarrow  \KK$,
 where $B^d \subseteq \HH$
 is the closed unit ball, 
 \[
      B^{d}:=\{x\in \HH: \norm{x}\leq 1\},
\]
such that, for all Parseval frames $\{x_j\}_{j\in J} \subseteq B^d$ for $\HH$, one has
 \[
      \sum_{j\in J} g(x_j)=W.
 \]
 \end{defn}

Our {\it goal} is the following:
{\it Define, 
implement, and generalize the notion of Gleason's functions for ONBs and the unit 
sphere to the setting of complex Parseval frames
and the closed unit ball.} It turns out that there are  fundamental
mathematical implications and new technology required to implement
Gleason's theorem in this setting.

The {\it reason} we shall pursue this goal is that 
a version of Gleason's theorem has been proved 
in the setting of positive operator valued measures (POVMs) \cite{busc2003}, \cite{CavFucManRen2004}
{\it and} POVMs can be viewed as equivalent to Parseval frames, a fact established and exploited in quantum detection problems \cite{BenKeb2008}, see Section \ref{sec:parsframpovm}.

A  {\it consequence} of this goal and reason is a quantitative insight into Busch's
formulation of Gleason's theorem in terms of his notion of a generalized probability measure,
see Section \ref{sec:appl}.

\begin{rem}[The Welch bound]
\label{rem:welch}
There are natural problems and relationships to be resolved
and understood. For example, it is not difficult  to check that if
$g$ is a Gleason function of weight $W_N$ for all unit norm frames with
$N > d$ elements for a given $d$-dimensional Hilbert space
$\HH = \KK^d$, then $g$ is constant on $S^{d-1}$.
On the other hand, 
we can formulate the definition of a Gleason function to 
consider the class of all {\it equiangular Parseval} frames, thereby interleaving the power
of Gleason's theorem with fundamental problems of equiangularity as they relate to
the Welch bound and optimal ambiguity function behavior, see Definition \ref{defn:frame} {\it c}
and Appendix \ref{sec:app}.
This is inextricably related
to the construction of constant amplitude finite sequences with $0$-autocorrelation,
whose narrow-band ambiguity function is comparable to the Welch bound, e.g.,
see \cite{BenBenWoo2012} and \cite{AndBenDon2019}.
\end{rem}

%%%%%%%%%%%%%%%%%%%

\subsection{Outline}
\label{sec:outline}
In Section \ref{sec:GT}, we summarize Gleason's work in \cite{glea1957}
in order to motivate further the definition and analysis of Gleason functions. 
We also extend his fundamental theorem (Theorem \ref{thm:glea1}) from the setting of 
non-negative functions to that of bounded functions, viz., Theorem \ref{thm:Gleasonbounded}. The proof is
elementary, but is necessary in the proof of our basic Theorem \ref{thm:GNsolution}.
Subsection \ref{sec:d=2} may seem superfluous to Gleason's observation that 
the difficult direction of his theorem fails for $d=2$, but we do provide a reason for why this is so
by a characterization of quadratic forms on $S^1$.

Section \ref{sec:parsframpovm} establishes the well-known relationship between POVMs and Parseval frames.
The former have long been a staple in quantum measurement, e.g., \cite{busc2003}, 
\cite{CavFucManRen2004};
and the latter is the central mathematical reason we have gone beyond Gleason's use of ONBs.

Sections \ref{sec:gleapars}  - \ref{sec:appl} establish our basic theory. The Parseval frame
formulation of POVMs allows us to look more deeply into Gleason functions in Sections 
\ref{sec:glean} and \ref{sec:appl}.

Section \ref{sec:gleapars} gives the basic properties of Gleason functions 
for the Parseval frames for $\HH = \KK^d$.
Theorem \ref{thm:gleapars} shows that
quadratic forms 
defined by self-adjoint operators are always Gleason functions for the Parseval frames for $\KK^d$, 
similar to the case
of Gleason functions for the ONBs for $\KK^d$.
We then prove that continuous or non-negative Gleason functions for the
Parseval frames for $\KK^d$ are 
reminiscent of homogeneous functions of degree 2 on $B^d$ (Theorems \ref{thm:continuoushomog}
and \ref{thm:nonnegativehomog}).
Using Theorem \ref{thm:nonnegativehomog},
we characterize bounded, real-valued Gleason functions for Parseval frames in terms of quadratic forms
defined by self-adjoint operators in analogy to results in \cite{glea1957}.
This is Theorem \ref{thm:realmainthm}, the converse of Theorem 
\ref{thm:gleapars}, cf., Theorem \ref{thm:complexmainthm}.

Because Parseval frames for $\HH = \KK^d$ vary in cardinality $N$, a natural question arises
about the relationship between the sets $\Gcal_N$ of Gleason functions for the $N$-element Parseval frames
for $\KK^d$ as $N$ varies. $\Gcal_N$ is the subject of Section \ref{sec:glean}. 
In Theorem \ref{thm:GNsolution}, we prove that if $N \geq d+2$, then
${\rm card}\,\Gcal_N ={\rm card}\,\Gcal_{N+1}$. The proof requires several propositions of independent interest.
In Section \ref{sec:appl}, we use Theorem \ref{thm:GNsolution} to weaken the hypotheses
in Busch's theorem, that itself is an analog of Gleason's 
mathematical characterization of quantum states. 

Since Parseval frames are central to our theory, and because they play an important role
in applications ranging from numerically effective noise reduction to the construction of 
Grassmannian frames dealing with spherical codes to geometrically uniform codes
in information theory to Zauner's conjecture in quantum measurement, we conclude with 
Appendix 
\ref{sec:app} 
putting some of these topics in context.

%%%%%%%%%%%%%%%%%%%%%%%%
%%%%%%%%%%%%%%%%%%%%%%%%%

\section{Gleason's theorem}
\label{sec:GT}

%%%%%%%%%%%%%%%%%%%%%%

\subsection{Preliminaries}
\label{sec:prelim}

In order to state Gleason's theorem, viz., Theorems \ref{thm:glea1} and \ref{thm:glea2}, we need the 
following set-up and notions.
Let $A: \HH \longrightarrow \HH$ be a linear operator, where $\HH$ is a separable Hilbert space
defined over $\KK$. 
We make the convention that $q(x) := \langle A(x), x \rangle$ is a {\it quadratic form} in the sense that 
$q(\alpha x) = |\alpha |^2 q(x)$ for all $x \in \HH$ and $\alpha \in  \KK$, see Remark \ref{rem:quadform}.
$A$ is {\it bounded}, i.e., {\it continuous}, if 
$\norm{A}_{op} := {\rm sup}_{\norm{x} \leq 1} \norm{A(x)}_{\HH}  < \infty$;
and ${\mathcal L}(\HH)$ denotes the space of bounded linear operators
$A: \HH \to \HH$. The {\it adjoint} $A^*$ of $A$ is the mapping $A^*: \HH \to \HH$ defined by the formula
$\langle A(x), y \rangle = \langle x, A^*(y) \rangle$ for all $x, y \in \HH$.
$A$ is {\it self-adjoint} if $A$ is bounded and $A^* = A$;
and $A$ is {\it positive}, resp.,
{\it positive
semi-definite} if
\[
    \forall x \in \HH \setminus \{0\}, \quad \langle A(x),x \rangle > 0, \, {\rm resp.}, \, \geq 0.
\]

Let $\Lcal_+(\HH)$  denote the subset
of positive semi-definite elements of $\Lcal(\HH)$; and let $\Scal_+(\HH)$
denote the set of positive semi-definite self-adjoint operators on $\HH$.

Recall that if $A: \HH \longrightarrow \HH $ is a linear operator on a Hilbert space 
$\HH$ defined over $\KK = \CC$ and  $\langle A(x),x \rangle \in \RR$ for all $x \in \HH$,
then $A$ is self-adjoint.  Conversely, if 
$A$ is self-adjoint, then
\begin{equation}
\label{eq:saimpliesAxxreal}
    \forall x \in \HH, \quad \langle A(x), x \rangle \in \RR.
\end{equation}
If $A$ is self-adjoint, then the eigenvalues $\lambda$ of $A$ are real. Thus,
in the case that $A$ is positive, resp.,
positive
semi-definite, then $\lambda >0$, resp., $\lambda \geq 0$.
When $\HH$ is defined over $\KK = \RR$
and $A$ is positive, resp., 
positive
semi-definite, we also have that $\lambda >0$, resp., $\lambda \geq 0$,
without having to verify that $A$ is self-adjoint.  

 If 
$\HH = \KK^d$, then we consider linear operators $A: \HH \longrightarrow \HH $ with
the $d \times d$ matrix $A = (a_{i,j})_{i,j = 1}^d$. $A$ is easily checked to be bounded by 
making the matrix calculation,
\[
    \norm{A(x)}_{\HH}  \leq \sum_{i,j=1}^d |a_{i,j}|^2 |x_j|^2 \leq    
    \big( \sum_{i,j=1}^d |a_{i,j} |^2 \big) \norm{x}_{\HH}^2.
\]

\begin{rem}[Quadratic forms]
\label{rem:quadform}
Classically, and differing from our convention in the case $\KK = \CC$, a
{\it quadratic form} over $\KK^d$ in the $d$ variables $x_1,x_2, \ldots, x_d \in \KK$ is a polynomial,
\begin{equation}
\label{eq:quad}
        Q(x) := q(x_1, \ldots, x_d)  := \sum_{i=1}^d\sum_{j=1}^d c_{i,j} x_i x_j, \quad c_{i,j} \in \KK,
\end{equation}
in which every term has degree $2$, i.e., every term is a multiple of $x_i x_j$
for some $i, j$. If we set $a_{i,j} := \frac{1}{2}(c_{i,j} + c_{j,i})$, and consider the matrices
$A =(a_{i,j})$ and $C =(c_{i,j})$, then $A$ is symmetric and 
\begin{equation}
\label{eq:symmquad}
     \forall x = (x_1, \ldots, x_d), \quad x^{\tau} A(x) 
     = \sum_{i=1}^d\sum_{j=1}^d a_{i,j} x_i x_j =\frac{1}{2} Q(x) +\frac{1}{2} Q(x) = Q(x),
\end{equation}
where $\tau$ denotes the transpose.
%see Subsection \ref{sec:homog} for more details on quadratic forms.
If $\KK = \RR$, then $x^{\tau} A(x)= \langle A(x), x \rangle$, where $x$ is a $d \times 1$ vector
in the matrix multiplication $A(x)$, cf. \eqref{eq:saimpliesAxxreal}.
This is not true for $\KK = \CC$ because of conjugation.
 \end{rem}
 
The {\it trace}, ${\rm tr}(A)$, of a $d \times d$ matrix $A =(a_{i,j})$ is 
\[
     {\rm tr}(A) := \sum_{j=1}^d \, a_{j,j}.
\]
For $d \times d$ matrices $A,\,B$, we have $a {\rm tr}(A) + b {\rm  tr}(B) =
{\rm  tr}(a A + b B), \, {\rm  tr}(AB) = {\rm  tr}(BA)$, and 
${\rm  tr}(A^{\ast}) = \overline{{\rm  tr}(A)}$,
where $A^{\ast}$ is the adjoint of $A$. Further, if $A$ is self-adjoint, or, more
generally, if $AA^{\ast}=A^{\ast}A$, i.e., $A$ is complex normal, then
\begin{equation}
\label{eq:trace}
      {\rm tr}(A) = \sum_{j=1}^d \, \lambda_j \quad {\rm and} 
      \quad {\rm tr}(A^{\ast}A) = \sum_{j=1}^d \, | \lambda_j|^2,
\end{equation}
where the $\lambda_j$ are the not necessarily distinct eigenvalues of $A$.

Given $\HH = \KK^d$. Let $B: \HH \longrightarrow \HH$ be a linear operator,  
let $\{e_1, \ldots, e_d \}$ be the standard ordered basis for $\HH$,
and let $M_B = (b_{i,j})$ be the $d \times d$ matrix representation of $B$
in this basis. ($\{e_1, \ldots, e_d \}$ {\it standard} means
that each $e_j = (0, \ldots, 0, 1,0, \ldots, 0)$, where $1$ is in the $j$-th coordinate.)
The {\it trace} of $M_B$, denoted by ${\rm tr}(M_B)$ is 
\[
     {\rm tr}(M_B) := \sum_{j=1}^d \, b_{j,j}.
\]

\begin{rem}[Trace class]
\label{rem:trace}
Let $\HH$ be a separable Hilbert space
defined over $\KK$. By definition,
$A \in {\mathcal L}(\HH)$ is a {\it trace class operator} if for some and hence all
ONBs $\{ x_n \}$ for $\HH$,
\[
     \norm{A}_1 :=  \sum_n \langle (A^*A)^{1/2} (x_n), x_n \rangle < \infty,
\]
so that $\sum_n \langle A(x_n), x_n \rangle < \infty$ when $A$ is self-adjoint, noting
that  $\langle (A^*A)^{1/2} (x_n), x_n \rangle \geq 0$. The {\it trace} of $A$ is
${\rm tr}(A) := \sum_n \langle A(x_n), x_n \rangle$, and this is compatible with
\eqref{eq:trace}.

Further, every {\it compact operator} $A \in \Lcal(\HH)$ is characterized by the representation,
\[
    \forall x \in \HH,\;   
    A(x)=\sum _{j} \lambda _{j} \langle x,y_{j} \rangle x_{j},
    \quad {\text{where}} \quad  \lambda _{j}\geq 0 \quad {\text{and}} \quad 
    \lambda _{j}\rightarrow 0, 
\]
for some orthonormal bases $\{x_j\}$ and $\{y_j\}$ for $\HH$. 

As is well-known, finite rank operators
$A \in \Lcal(\HH)$ are trace class, and these are Hilbert-Schmidt, and these are compact. We mention this
since the dual of the space of compact operators with the proper topology is the space of trace class operators, 
and because of Theorem \ref{thm:spectral}b, see \cite{rudi1991} for all of this material.
\end{rem} 

We shall use the spectral theorem several times throughout, and state the following form, see
\cite{GohGol1981}, \cite{rudi1991}, \cite{stra1988}, \cite{lay-1994}, \cite{FriInsSpe1997}, \cite{TreBau1997}.

\begin{thm}
\label{thm:spectral}
{\it a.} Let $\HH = \KK^d$, let $A: \HH \longrightarrow \HH$ be a linear operator, and 
for convenience denote $M_A$ by 
$A$. If $A$ is self-adjoint, i.e., $A$ a real symmetric matrix if $\KK = \RR$ or an Hermitian matrix if
$\KK = \CC$, then there exists a matrix $U$ with columns consisting of a complete set of
orthonormal eigenvectors for $A$, such that $\Lambda 
= UA U^{-1}$ is diagonal. Such a $U$ is orthogonal if $\KK = \RR$ and unitary if $\KK = \CC$.
%$A$ can be diagonalized by $U$ as the diagonal matrix $\Lambda 
%= U^{-1} A U$, where the columns of $U$ form a complete set of orthonormal eigenvectors,
%and where $U$ is orthogonal if $\KK = \RR$ and unitary if $\KK = \CC$.

{\it b.} Let $\HH$ be a separable Hilbert space defined over $\KK$, and let $A \in \Lcal(\HH)$
be a compact self-adjoint operator. There is an orthonormal sequence $\{x_j \} \subseteq \HH$
of eigenvectors of $A$ and a corresponding sequence $\{ \lambda_j \} \subseteq \KK$ of eigenvalues, such that
\[
   \forall x \in \HH, \quad A(x) = \sum_{j}\, \lambda_j \langle x, x_j \rangle\, x_j.
\]
If $\{ \lambda_j \}_{j=1}^{\infty} $ is an infinite sequence, then ${\rm lim}_{j \to \infty}\, \lambda_j = 0$.
\end{thm}

\begin{rem}[Spectral decomposition]
\label{rem:spectral}
{\it a.} With regard to part {\it a} of Theorem \ref{thm:spectral}, we note the following.
In the real symmetric case, we have $A = U \Lambda U^{-1}$, with orthonormal eigenvectors 
forming $U$
and with the eigenvalues of $A$ forming the diagonal matrix $\Lambda$. Also, in the $\KK = \CC$ case 
the eigenvalues of self-adjoint $A$ are real; and, in both cases, if two eigenvectors
 come from distinct eigenvalues, then they are orthogonal.
 
 Further, to prove the existence of an ONB of eigenvectors in the case that 
 $\KK = \CC$ and $A$ is Hermitian, we apply the 
 fundamental theorem of algebra to the characteristic polynomial of $A$ to obtain an eigenvalue
 $\lambda_1$ and an eigenvector $u_1$. Then, we consider the orthogonal complement
 of $u_1$ to obtain a $u_2$, and, continuing in this way, we see how to construct
 a complete set of orthonormal eigenvectors. The case $\KK = \RR$ can be deduced
 from the complex case by complexification, see, e.g., \cite{halm1958}, Section 77.
 
 {\it b.} With regard to part {\it b} of Theorem \ref{thm:spectral}, we note the following. Although
 any linear operator on $\HH = \KK^d$ has an eigenvalue, that is not necessarily the case even
 for self-adjoint operators on infinitely dimensional $\HH$. Further, the self-adjoint identity operator $I$
 on infinite dimensional $\HH$ is {\it not} a compact operator.
 
\end{rem}
%%%%%%%%%%%%%%%%%%%%

\subsection{Gleason's theorem}
\label{sec:glea}

If $\KK = \CC$, then the following holds for normal operators $A$.

\begin{thm}
\label{thm:saimpliesglea}
 Let $\HH = \KK^d$ and let $A$ be a self-adjoint linear operator $A:\HH \to \HH$.
 The function $g: \HH \to \KK$, defined by the formula,
 \begin{equation}
\label{eq:glea1}
       \forall x\in S^{d-1}, \quad g(x)=\langle A(x),x \rangle,
\end{equation}
is a Gleason function of weight $W = {\rm tr}(A)$ for the ONBs for $\HH$.
\end{thm}

\begin{proof}
By the spectral
theorem, there exists an orthonormal eigenbasis $\{e_j\}_{j=1}^d$
associated with the set $\{ \lambda_j \}_{j=1}^d$ of eigenvalues 
of $A$. Hence, for all $x\in \HH$ we have $x=\sum_{j=1}^d \inner{x}{e_j}e_j$ and 
$A(x)=\sum_{j=1}^d \inner{x}{e_j}\lambda_j e_j.$
If $\{x_j\}_{j=1}^d$ is an ONB for $\HH$, then we have
\[
       {\rm tr}(A)=\sum_{j=1}^{d}\lambda_j
=\sum_{j=1}^d \lambda_j\norm{e_j}^2
=\sum_{j=1}^d \lambda_j\sum_{n=1}^d |\inner{e_j}{x_n}|^2,
\]
where the last equality is due to 
the Parseval identity.
Reordering the finite sums, and using the orthogonality
of $\{e_j\}$ yields the desired result:
\[
         {\rm tr}(A) = \sum_{n=1}^d\sum_{j=1}^d \lambda_j\inner{e_j}{x_n}\inner{x_n}{e_j}
 =\sum_{n=1}^d\sum_{j=1}^d\inner{\inner{x_n}{e_j}\lambda_j e_j}{x_n}
      =\sum_{n=1}^d\inner{A(x_n)}{x_n}=\sum_{n=1}^d g(x_n).
\]
Therefore, $g$ is a Gleason function of weight $W = {\rm tr}(A)$ for the ONBs for $\HH$.
\end{proof}

The converse assertion of Theorem \ref{thm:saimpliesglea} is true directly for $d=1$. 
In fact, if $d=1$ and $g$ is a Gleason function of weight $W$ for the two ONBs for $\HH = \KK = \RR$,
then $A$ is defined by the action $A(x) := Wx$. The same operator works for  $\HH = \KK = \CC$, but in this case
the ONBs are the uncountable set, $\{z_u = e^{iu} : u \in [0, 2\pi) \}$, and $A$ is again defined by the 
action $A(x) :=Wx$.

The converse assertion of Theorem \ref{thm:saimpliesglea} is not true 
for the case $d=2$, see Subsection \ref{sec:d=2}.

Although the situation is substantially more intricate for $d \geq 3$, 
Gleason's Theorem \ref{thm:glea1} asserts that the converse 
of Theorem \ref{thm:saimpliesglea} is still true, but with restrictions on the 
given Gleason function. As Gleason was well aware, some restrictions are necessary, see Proposition
\ref{prop:disc}.

\begin{thm}
\label{thm:glea1}
Let $\HH =\KK^d$ and let $g:S^{d-1} \longrightarrow \RR$ 
be a non-negative Gleason function for the ONBs for $\HH$, where $d \geq 3$. There
exists a positive self-adjoint operator $A:\HH \longrightarrow \HH$ such that
\[
       \forall x\in S^{d-1}, \quad g(x)=\langle A(x),x \rangle.
\]
The result is also true for any separable Hilbert space $\HH$.
\end{thm}

\begin{rem}[Gleason's theorem for $\RR^3$]
\label{rem:glear3}
The proof of Theorem \ref{thm:glea1} depends on Gleason's theorem that 
{\it non-negative Gleason functions for the ONBs for $\RR^3$ 
satisfy \eqref{eq:glea1}} (\cite{glea1957}, Theorem 2.8); and this, in turn,
depends on his result that {\it continuous Gleason functions for the ONBs for $\RR^3$ 
satisfy \eqref{eq:glea1}}
(\cite{glea1957}, Theorem 2.3).
Both proofs are ingenious.
\end{rem}

 Theorem \ref{thm:glea1} is essential and significant for the proof of the following result.
  The positivity hypothesis 
 in Theorem \ref{thm:glea1} is natural given the measure theoretic nature of Theorem
\ref{thm:glea2}. Theorem \ref{thm:glea2} was our starting point in Subsection \ref{sec:back}.

\begin{thm}
\label{thm:glea2}
Let $\mu$ be a measure on the 
closed subspaces of $\HH$, where  ${\rm dim} (\HH) \geq 3$. There
exists a positive semi-definite self-adjoint trace class operator 
$A:\HH \longrightarrow \HH$ such that, for all
closed subspaces $X\subseteq \HH$, 
\[
       \mu(X)={\rm tr}(AP_X),
\] 
where $P_X$ is the
orthogonal projection of $\HH$ onto $X$. 
\end{thm}

\begin{proof} Let $B_x=\overline{\rm span}(x)$ for any unit norm vector
$x \in \HH$, i.e., $x \in S^{d-1}$ for $\HH = \KK^d$. Then,
$g(x) = \mu(B_x)$ defines a 
non-negative Gleason function for the ONBs
for $\HH$ by the definition of $\mu$. By Theorem \ref{thm:glea1}, there exists a positive
self-adjoint operator $A$ such that for all unit norm 
$x \in \HH$, we have $g(x)=\inner{A(x)}{x}$. 

Next, note that if  $\{x_j\}$ is an ONB for $\HH$, then 
\[
     \mu(\HH)=\sum_j \mu(B_{x_j})=\sum_j\inner{A(x_j)}{x_j}={\rm tr}(A),
\]
where the sums are finite since, by the definition of a measure $\mu$ on the closed subspaces
in Subsection \ref{sec:back},
we have assumed $\mu(\HH) < \infty$. Because the latter sum is finite, $A$ is trace class
and, in fact, ${\rm tr}(A) = \mu(\HH)$. These latter assertions are immediate for the
cases that $\HH = \KK^d, \, d \geq 3$.

If $X \subseteq \HH$ is an arbitrary closed subspace, choose an ONB $\{y_i\}$ for $X$,
and an ONB $\{z_j\}$ for the orthogonal complement $X^{\perp}$ of $X$. Then, the
projection mapping $P_X$ satisfies
$P_X (y_i)=y_i$ and $P_X( z_j)=0$ for all $i$ and $j$. Clearly, $\{y_i\}\cup \{z_j\}$ is an 
ONB for $\HH$. Therefore, we have
\[
        \mu(X)=\sum_i \mu(B_{y_i})=\sum_{i}\inner{A(y_i)}{y_i}
        =\sum_{i} \inner{A(P_X (y_i))}{y_i}+\sum_{j}\inner{A(P_X(z_j))}{z_j}={\rm tr}(AP_X),
\]
as desired. 
\end{proof}

Whereas Gleason formulated Theorem \ref{thm:glea1} only for the case
where $g$ takes on non-negative real values, we now show that it is not
difficult to extend the result to the more general case that $g: S^{d-1} \to \KK$
is bounded. (Here, $\KK$ is the base field of $\HH = \KK^d$.) In fact, this
generality is used in the sequel, e.g., in Theorem \ref{thm:GNsolutionGof0}.

\begin{thm} 
\label{thm:Gleasonbounded}
Let $\HH = \KK^d$ and let $g: S^{d-1} \to \KK$ be a 
bounded Gleason function for the ONBs for $\HH$, where $d \geq 3$. There exists a 
bounded (necessarily since $\HH = \KK^d$) linear operator $A: \HH \to \HH$ such that
\[
\forall x \in S^{d-1}, \hspace{2em} g(x) = \inner{A(x)}{x}.
\]
Furthermore, $A$ is self-adjoint if the bounded function $g$ is real-valued,
and, in particular, if $\HH = \RR^d$.

\end{thm}

\begin{proof}
{\it i.} First suppose that the image of $g$ lies in $\RR$. Let $W$ denote the weight of $g$ and 
let $\lambda = \inf_{x \in S^{d-1}} g(x)$. Then, the function $f: S^{d-1} \to \KK$ defined by $f(x) := g(x) - \lambda$ is a Gleason function of weight $W - \lambda d$ for the ONBs for $\HH$, since if $\{ x_i \}_{i=1}^d$ is an ONB for $\HH$, then
\[
\sum_{i=1}^d f(x_i) = \Big(\sum_{i=1}^d g(x_i)\Big) - \lambda d = W - \lambda d.
\]
Furthermore, $f$ is non-negative, since $f(x) \geq \lambda - \lambda = 0$ for all $x \in S^{d-1}$. 
Hence, Gleason's Theorem \ref{thm:glea1} implies that there exists a self-adjoint operator 
$B: \HH \to \HH$ such that $f(x) = \inner{B(x)}{x}$ for all $x \in S^{d-1}$. 
Setting $A := B - \lambda I$, we obtain 
$g(x) = \inner{(B - \lambda I)(x)}{x} = \inner{A(x)}{x}$ for all $x \in S^{d-1}$. Note that $A$ is 
the difference of two self-adjoint operators and is therefore self-adjoint.

{\it ii.} Now we proceed to the general case, where the image of $g$ lies in $\CC$. 
Both ${\rm Re}\, g$ and ${\rm Im}\, g$ are Gleason functions for the ONBs for $\HH$, since for any ONB 
$\{ x_i \}_{i=1}^d$ for $\HH$, we have
\[
  \sum_{i=1}^n {\rm Re}\, g(x_i) + i \sum_{i=1}^n {\rm Im}\, g(x_i) = \sum_{i=1}^n g(x_i) = W = 
  {\rm Re}\, W + i {\rm Im}\, W.
\]
Clearly, both ${\rm Re}\, g$ and ${\rm Im}\, g$ are bounded. Hence, by part {\it a} we obtain linear operators 
$B, C: \HH \to \HH$ such that ${\rm Re}\, g(x) = \inner{B(x)}{x}$ and ${\rm Im}\, 
g(x) = \inner{C(x)}{x}$ for all $x \in S^{d-1}$. 
Setting $A := B + iC$, we obtain $g(x) = \inner{A(x)}{x}$ for all $x \in S^{d-1}$.
\end{proof}

Although $B + i C$ is not self-adjoint, i.e., $\langle (B+iC)(x),y \rangle \neq 
\langle x, (B^\ast + i C^\ast)(y) \rangle$, we do have $\langle (B+iC)(x),y \rangle =
\langle x, (B^\ast - i C^\ast)(y) \rangle$.

\begin{example}[Gleason functions as compositions]
\label{ex:comp}
{\it a.} Let $g$ be the composition $g = f \circ h$, 
restricted to $S^{d-1}$, of $h : \KK^d \to\KK$ and $ f :\KK \to \KK$,
and where $h$ takes an arbitrary constant value $c \in \KK$ on $S^{d-1}$.
Then, for any ONB
$\{x_j\}_{j=1}^d\subset \KK^d$, we have
 \[
      \sum_{j=1}^{d} g(x_j)=\sum_{j=1}^{d} (f \circ h)(x_j)=\sum_{j=1}^d f(c)= d \cdot f(c) :=W.
\]
Thus, $g$ is Gleason function of weight $W$ for the ONBs for $\KK^d$.
We can write $g$ as 
 \[
      \forall x\in S^{d-1}, \quad g(x)=\inner{A(x)}{x},
\]
 where $A=\frac{W}{d}I$ and where $I:\KK^d \to \KK^d$ is the identity. In fact,
 for $x \in S^{d-1}$, we have
 \[
    \langle A(x),x \rangle =\frac{W}{d} \langle x,x \rangle = f(c) = (f \circ h)(x) = g(x).
 \]
Note that $A$ is self-adjoint if $f : \KK \to \RR$. Further,
in this case, 
if $\KK = \RR$, then $y^T A(y)$ is the quadratic form $(W/d)(y_1^2 + \ldots + y_d^2)$.

It is natural to consider the special 
case $c = \norm{x} =1$ since $x \in S^{d-1}$.

{\it b.} Let $g$ be the composition $g = f \circ h$, where $h: S^{d-1} \to \KK$
is a Gleason function of weight $W$ for the ONBs for $\KK^d$ and where
$f : \KK \to \KK$ is a homomorphism on the additive group $\KK$. For any ONB
$\{x_j\}_{j=1}^d\subset \KK^d$, we have
\[
  \sum_{j=1}^d g(x_j) = f(h(x_1)) + \ldots + f(h(x_d)) = f(h(x_1) + \ldots + h(x_d)) = f(W),
\]
and so $g$ is a Gleason function of weight $f(W)$ for the ONBs for $\KK^d$.

{\it c.} Let $f : \RR \to \RR$ be a homomorphism on the additive group $\RR$. Since
$f(0)= 0$, and setting $c := f(1) \in \RR$, we obtain $f(q) = cq$ for all $q \in \QQ$ by direct calculation.
If $f$ is continuous on $\RR$, we can then assert that $f(x) = cx$ for all $x \in \RR$.
As is well-known, the hypothesis of continuity can be relaxed to assuming only that
$f$ is continuous at a point or even only that $f$ is Lebesgue measurable, and
one still verifies that  $f(x) = cx$ for all $x \in \RR$.
\end{example}

$\RR$ is an infinite dimensional vector space over the rational field $\QQ$. For this
setting, we say that $H$ is a {\it Hamel basis} for $\RR$ if
\[ 
    \forall x \in \RR, \; \exists \{r_{\alpha}\} \subseteq \QQ \; {\rm and} \; \exists
    \{h_{\alpha}\} \subseteq H, \; \text{such that} \, x = \sum_{\alpha} r_{\alpha}h_{\alpha},
\]
where the sum is finite and the representation is unique. Using Zorn's lemma, 
which is an equivalent form of the axiom of choice, it is straightforward to 
see that Hamel bases exist by the following argument: let $\Ical$ be the family of all
subsets of $\RR$ that are linearly independent over $\QQ$; then there is a maximal element
$H \in \Ical$, and this $H$ can be shown to be a Hamel basis 
$\RR$. We have ${\rm card}(H) = {\rm card}(\RR)$. Further, any vector space
over any field has a Hamel basis over that field. See \cite{BenCza2009}, pages 88--89, 150, 153, 162
for this material, its relation to measure theory, and classical references beyond
Hausdorff's fundamental book. 

\begin{prop}
\label{prop:disc}

{\it a.} There are discontinuous homomorphisms $f : \RR \to \RR$.

{\it b.} There are discontinuous, and, in fact, non-Lebesgue-measurable Gleason functions for the
ONBs for $\HH = \KK^d$.

{\it c.} There are Gleason functions $g$ for the ONBs for $\HH = \KK^d$ that do not satisfy
\eqref{eq:glea1} for any self-adjoint operator $A$.

\end{prop}

\begin{proof}
{\it a.} Define $f(t) \in \RR$ for each $t \in H$, a Hamel basis for $\RR$ over  $\QQ$.
This definition of $f$ on $H$ is arbitrary.
Each $u \in \RR$ has a unique finite sum representation $u = \sum_{t \in H} r_t(u) t$, where
$r_t(u) \in \QQ$. Define $f : \RR \to \RR$ by $f(u) = \sum_{t \in H} r_t(u) f(t)$. $f$ is
well-defined on $\RR$ since $H$ is a Hamel basis. By the unique representations,
$u = \sum_{t \in H} r_t(u) t$, $v = \sum_{t \in H} r_t(v) t$, $u + v = \sum_{t \in H} r_t(u +v) t$,
we can assert that $f$ is a homomorphism because
\[
    \forall u,v \in \RR, \quad  r_t(u) + r_t(v)  = r_t(u +v),
\]
and so 
\begin{equation}
\label{eq:homom}
   f(u) + f(v) = \sum_{t \in H} r_t(u) f(t) + \sum_{t \in H} r_t(v) f(t)  
   = \sum_{t \in H}( r_t(u) + r_t(v)) f(t) = \sum_{t \in H} r_t(u +v) f(t) = f(u + v).
\end{equation}
The facts that ${\rm card} (H) = {\rm card}(\RR) = {\bf c}$ and $f$ can be
defined arbitrarily on $H$ allow us to conclude that the
homomorphism equation \eqref{eq:homom} has ${\bf c}^{\bf c}$ solutions.

On the other hand, any of these solutions $f$ that is continuous at a point
is of the form $f(u) = cu$ for some $c \in \RR$ (Example \ref{ex:comp} c),
i.e., there are {\bf c} such solutions, and so all of the other solutions are discontinuous.

{\it b.} Combining Example \ref{ex:comp} b with part {\it a} gives part {\it b}.

{\it c.} Suppose $g$ is a Gleason function that satisfies \eqref{eq:glea1}.
The continuity of $g$ on the compact set $S^{d-1}$ follows since $A$ is bounded.
More concretely, if $\norm{x_n - x}_{\HH} \to 0$, where $x_n, x \in S^{d-1}$, then
\[
    |g(x_n) - g(x)| = |\langle A(x_n - x), x_n \rangle + \langle A(x), x_n \rangle - \langle A(x),x \rangle | 
\]
\[
    \leq \norm{A(x_n - x)}_{\HH} + |\langle x_n - x, A(x) \rangle| \leq 2 \norm{A} \norm{x_n - x}_{\HH},
\]
which goes to $0$ in the limit.

Choose a discontinuous Gleason function from part {\it b}. It can not satisfy  \eqref{eq:glea1} for
then it would be continuous.
\end{proof}
 
%%%%%%%%%%%%%%%%%%%%%%%%

\subsection{The case $d=2$}
\label{sec:d=2}

Although, as noted in Subsection \ref{sec:glea}, the converse assertion of Gleason's 
Theorem \ref{thm:saimpliesglea} is elementary to verify in the 1-dimensional case and is true
 for separable
Hilbert spaces of dimension $d \geq 3$, the theorem does not hold for $\HH = \KK^2$.
This means that there are Gleason functions $g$ for the ONBs for $\KK^2$ for which there
are no self-adjoint operators $A_g :\KK^2 \to \KK^2$ with the property that 
$g(x) = \langle A_g(x),x \rangle$ on $S^1$. Our only insight about this
assertion and that Gleason did not explicitly make is that our proof of Proposition
\ref{prop:Gleason 2d} requires the characterization of quadratic forms over $\RR^2$
that we give in Proposition \ref{prop:Quadratic form zeros}.
Further, the fact 
remarked by Gleason \cite{glea1957}, page 886, that, on $\RR^2$
for example, Gleason functions for the ONBs for $\RR^2$ can be defined arbitrarily
on the quadrant $\theta \in [0, \pi/2)$ of $S^1$ is routinely quantified in Example \ref{ex:2}
and Proposition \ref{prop:glearem}.

\begin{example}[Counterexample for converse in $\RR^2$]
\label{ex:2}
Define the often used $0,1$-valued function,
\[
      g(\theta)=\begin{cases} 
1, & \theta\in [0,\pi/2)\cap\QQ,\\
0, & \theta\in [0,\pi/2)\cap\QQ^c,\\
1-g(\theta-\pi/2), & \theta\in [\pi/2,\pi),\\\
g(\theta-\pi), & \theta\in [\pi, 2\pi),
\end{cases}
\]
where $\QQ^c$ is the set of irrational numbers and where
$\theta$ is the angle that a unit vector $u\in \RR^2$ takes with the positive $x$ axis in $\RR^2$.
Since every ONB for $\RR^2$ is of the form 
$$\{({\rm cos}\, (\theta),{\rm sin}\, (\theta)), (-{\rm sin}\, (\theta),{\rm cos}\, (\theta)) \},$$
for some $\theta \in [0,2\pi)$,
an elementary calculation shows that $g$ is a Gleason function for the ONBs
for $\RR^2$.
On the other hand, $g$ is clearly not a quadratic form since it is not continuous, see
\eqref{eq:symmquad}. Thus, Gleason's theorem does not extend to $2$-dimensional
real inner product spaces. For the more difficult case of a continuous counterexample
for the converse, see Proposition \ref{prop:Gleason 2d}.

\end{example}

The following result, a modified form of which was used by Gleason in his original paper, 
viz., his Lemma 2.2,
illustrates why Theorem \ref{thm:saimpliesglea} does not generalize to such spaces.

\begin{prop}
\label{prop:Gleason cos lemma}
Let $n \equiv 2 \mod 4$. The function $g(\theta) = 1 + \cos(n \theta)$,
defined on the unit circle $S^1 \subseteq \RR^2$, i.e., the polar coordinate $\theta \in [0,2\pi)$, 
is a non-negative Gleason function 
of weight $2$
for the ONBs for $\RR^2$.
\end{prop}
\begin{proof}
As noted in Example \ref{ex:2}, two unit vectors in $\RR^2$ form an orthonormal basis for $\RR^2$ if and only
if they have an angle of $\pi / 2$ radians between them. Indeed, this latter condition is 
equivalent to the orthogonality of the two vectors by the definition of angles in inner product 
spaces; and any two orthogonal unit vectors in $\RR^2$ form an ONB for $\RR^2$, 
since non-zero orthogonal vectors are linearly independent. 

Clearly, $1 + \cos (n \theta)$ takes on values in the range $[0, 2]$, so it remains to show that 
\begin{equation*}
        g(\theta_1) + g(\theta_2) = 2
\end{equation*}
for any $\theta_1, \theta_2$ giving the angles relative to the origin of the vectors in an 
ONB for $\RR^2$. Reordering if necessary, we can assume without loss of generality  
that $\theta_2 = \theta_1 + \pi / 2$. Using the trigonometric identity 
$\cos(\alpha + \beta) = \cos\alpha \cos\beta - \sin\alpha \sin\beta$ and the fact that $n \pi / 2$ 
is an odd multiple of $\pi$, we obtain
\begin{align*}
          g(\theta_1) + g(\theta_2) &= 2 + \cos(n \theta_1) + \cos(n \theta_1 + n \pi / 2) \\
         &= 2 + \cos(n \theta_1) + \cos(n \theta_1) \cos(n \pi / 2) - \sin(n \theta_1) \sin(n \pi / 2) \\
         &= 2 + \cos(n \theta_1) - \cos(n \theta_1) = 2
\end{align*}
for all angles $\theta_1$.
\end{proof}

In 
Proposition \ref{prop:Gleason 2d} we shall show that the Gleason functions defined in 
Proposition \ref{prop:Gleason cos lemma} are not 
quadratic forms on the unit circle $S^1 \subseteq \RR^2$ when $|n| \neq 2$.  
To this end we shall use the following result.

\begin{prop} 
\label{prop:Quadratic form zeros}
Any quadratic form over $\RR^2$ with more than four zeros on the unit circle is identically zero.
\end{prop}
\begin{proof}
Let $Q(v) = \inner{A(v)}{v}$ be a quadratic form over $\RR^2$ with more than four zeros
on the unit circle. Then, there exist unit vectors $v_1, v_2$ such that $Q(v_1) = Q(v_2) = 0$, 
but $v_2 \neq \pm v_1$. Thus, $v_1$ and $v_2$ are linearly independent, 
so they form a basis for the 2-dimensional space $\RR^2$. Consider the linear transformation 
$T$ defined over the standard basis for $\RR^2$ by $Te_i = v_i$ for $i = 1, 2$. 
$T$ is invertible because it sends a basis for $\RR^2$ to a basis for $\RR^2$. 

Define a new quadratic form $Q_2(v) :=  Q(T(v)) = \inner{AT(v)}{T(v)}$. We shall
 relate $Q_2$ to $Q$. Let $u$ be a unit vector in $\RR^2$. If $Q(u) = 0$,
  then $Q_2(T^{-1}(u) / \norm{T^{-1}(u)}) = 0$. Likewise, if $Q_2(u) = 0$, then 
  $Q(T(u)/ \norm{T(u)}) = 0$. The correspondences 
  $u \mapsto T^{-1}(u)/ \norm{T^{-1}(u)}$ and $u \mapsto T(u) / \norm{T(u)}$ 
 give inverse automorphisms of the unit circle. Indeed, for any unit vector $u$,
\begin{equation*}
     \frac{T^{-1}(T(u) / \norm{T(u)})}{\norm{T^{-1}(T(u) / \norm{T(u)})}} 
     = \frac{\norm{T(u)}}{\norm{T(u)}} \frac{T^{-1}T(u)}{\norm{T^{-1}T(u)}} = \frac{u}{\norm{u}} = u,
\end{equation*}
and, likewise,
\begin{equation*}
       \frac{T(T^{-1}(u) / \norm{T^{-1}(u)})}{\norm{T(T^{-1}(u) / \norm{T^{-1}(u)})}} = u.
\end{equation*}
In particular, the sets of zeros of $Q$ and $Q_2$ on the unit circle are in bijective correspondence, 
and, if either $Q$ or $Q_2$ is identically zero, then the other is as well.

The quadratic form $Q_2$ can be expressed in rectangular coordinates as 
$Q_2(x, y) = ax^2 + bxy + cy^2$ for some constants $a, b, c \in \RR$. 
Since $Q_2(e_1) = Q_2(e_2) = 0$ 
it follows that $a = c = 0$, and so $Q_2(x, y) = bxy$. If $b \neq 0$ then $Q_2$ is only zero 
on the unit circle at the values $\pm e_1, \pm e_2$, contradicting the hypothesis that $Q$, 
and hence $Q_2$, has more than four zeros. Thus, $b = 0$ and $Q_2$ is identically zero. 
By the previous comments, this implies that $Q$ is identically zero.
\end{proof}

There exist quadratic forms that have exactly four zeros on the unit circle, e.g., the 
quadratic form $Q$ defined in rectangular coordinates by $Q(x, y) = xy$. Hence, the 
hypothesis of {\it more than four zeros} in Proposition \ref{prop:Quadratic form zeros} cannot be further relaxed.

\begin{prop} 
\label{prop:Gleason 2d}
Let $n \equiv 2 \mod 4$ with $|n| \neq 2$. The function $g(\theta) = 1 + \cos (n \theta)$, 
defined on the unit circle $S^1 \subseteq \RR^2$, i.e.,
the polar coordinate $\theta \in [0,2\pi)$, is a Gleason function of weight $2$ for the 
ONBs for $\RR^2$, but it is not the restriction of a quadratic form to $S^1$.
\end{prop}
\begin{proof}
The fact that $g$ is a Gleason function is the content of Proposition \ref{prop:Gleason cos lemma}. 
It remains to show that $g$ is not the restriction of a quadratic form to the unit circle. 
Suppose that $g$ is
 a quadratic form on the unit circle. 
 Noting that $|n| = 6, 10, 14, \ldots$, we easily check that $g(\theta) = 1 + \cos (n \theta)$
 has at least $|n| \geq 6 > 4$ distinct zeros at $\theta = k\pi/|n|$, for
 $1 \leq k \leq 2|n| - 1$ and $k$ odd. Thus, $g$ is identically $0$ over $\RR^2$
 by Proposition \ref{prop:Quadratic form zeros}. We obtain the desired
 contradiction since $g(0) =1 + {\rm cos} (0) = 2 \neq 0$.
 
\end{proof}

The hypothesis $|n| \neq 2$ is necessary in Proposition \ref{prop:Gleason 2d}.
In fact, using the double-angle and 
Pythagorean trigonometric identities, the function 
$g(\theta) = 1 + \cos(\pm 2\theta)$ can be rewritten as 
\begin{equation}
\label{eq:quad2}
        g(\theta) = 1 + \cos^2(\pm \theta) - \sin^2(\pm \theta) = 2\cos^2(\pm \theta).
\end{equation}
Viewed in rectangular coordinates $(x,y)$ for inputs lying on the unit circle, 
the right side of \eqref{eq:quad2} is 
$2x^2$, which is a quadratic form.

The following is a quantitative version of Gleason's remark noted at the beginning of
this subsection.

\begin{prop}
\label{prop:glearem}
Let $f : \RR \to \RR$ be a bounded, non-negative, $\pi/2$-periodic
function, and let $W \geq {\rm sup} f$. Define $g$ in polar coordinates $\theta$ on $S^1$
by the formula
\begin{equation*}
     g(\theta) = \begin{cases}
    f(\theta) & \theta \in [0, \pi / 2) \cup [\pi, 3\pi / 2) \\
     W - f(\theta) & \theta \in [\pi / 2, \pi) \cup [3\pi / 2, 2\pi).
\end{cases}
\end{equation*}
Then, $g$ is a non-negative Gleason 
function of weight $W$ for the ONBs for $\RR^2$. 
\end{prop}

\begin{proof}
Since $f$ is non-negative and $W \geq \sup f$, it follows that $g$ is non-negative
on $S^1$.

As also noted in the proof of Proposition \ref{prop:Gleason cos lemma} and by the definition of
angles in inner product spaces,
 two non-zero vectors in $\RR^2$ 
are orthogonal if and only if they are separated by an angle of $\pi / 2$. Thus, it suffices to 
show that $g(\theta) + g(\theta + \pi / 2) = W$ for any angle $\theta$. 
If $\theta \in [0, \pi / 2) \cup [\pi, 3\pi / 2)$, then, taking angles modulo $2\pi$ as necessary, 
$\theta + \pi / 2 \in [\pi / 2, \pi) \cup [3\pi / 2, 2\pi)$; and, consequently,
\begin{equation*}
      g(\theta) + g(\theta + \pi / 2) = f(\theta) + W - f(\theta + \pi / 2) = f(\theta) + W - f(\theta) = W.
\end{equation*}
Otherwise, $\theta \in [\pi / 2, \pi) \cup [3\pi / 2, 2\pi)$, and, again taking angles modulo $2\pi$ as 
necessary, $\theta + \pi / 2 \in [0, \pi / 2) \cup [\pi, 3\pi / 2)$. Hence,
\begin{equation*}
     g(\theta) + g(\theta + \pi / 2) = f(\theta) + C - f(\theta + \pi / 2) = f(\theta) + W - f(\theta) = W.
\end{equation*}
Thus, $g$ is a Gleason function of weight $W$ for the ONBs for $\RR^2$ with weight $W$.
\end{proof}

%%%%%%%%%%%%%%%%%%%%%%%%%%%
%%%%%%%%%%%%%%%%%%%%%%%%%%%%

\section{Parseval frames and POVMs}
\label{sec:parsframpovm}

%%%%%%%%%%%%%%%%%

\subsection{Properties of frames}
\label{sec:framprop}

The following definition 
for Hilbert spaces is equivalent to Definition \ref{defn:framekd} for frames 
for ${\mathbb K}^d,$
and is formulated in terms of bounds that are often useful in computation
and coding.

\begin{defn}[Frames]
\label{defn:frame}
{\it a.} Let $\HH$ be a separable Hilbert space over the field $\KK$,
where $\KK = \RR$ or $\KK = \CC,$ 
e.g., 
$\HH = L^2({\mathbb R}^d), {\mathbb R}^d, {\mathbb C}^d.$
A finite or countably infinite sequence, 
$X = \{x_h\}_{h \in J},$ of elements of $\HH$ 
is a {\it frame} for $\HH$ if 
\begin{equation}
\label{eq:frame}
      \exists A, B > 0  \; \text{such that}  \;
      \forall x \in \HH, \quad A\norm{x}^2 \leq \sum_{h \in J} |\langle {x},{x_h}\rangle|^2 \leq B\norm{x}^2.
\end{equation}
The optimal constants, viz., the supremum over all such $A$ and infimum over all such $B$, are 
the {\it lower} and {\it upper frame bounds} respectively. When we refer to {\it frame bounds} 
or {\it constants} $A$ and $B$, we shall mean these optimal constants. Otherwise, we use the terminology,
{\it a lower frame bound} or {\it an upper frame bound}.

{\it b.} A frame $X$ for $\HH$ is an {\it A-tight frame} if $A = B.$ If a tight frame has the further property 
that $A = B = 1,$ then the frame is a {\it Parseval frame} for ${\mathbb H}.$   
A tight frame $X$ for ${\mathbb H}$ is a {\it unit norm tight frame} 
if each of the elements 
of $X$ has norm $1.$  
Finite unit norm tight frames for finite dimensional ${\mathbb H}$
are designated as FUNTFs. ONBs are both Parseval frames and FUNTFs 
for finite dimensional ${\mathbb H}$.

{\it c.} A set $X = \{ x_j \}_{j=1}^N \subset 
\HH = \KK^d$ is 
{\it equiangular} if 
\[  
 \exists \,  \alpha \geq 0 \quad {\rm such \; that} \quad \forall j \neq k, \quad |\langle x_j, x_k \rangle| = \alpha.
\]
An equiangular tight frame is designated as an ETF.
It is well known and elementary to verify that, for any $d \geq 1$, the simplex consisting of $N=d+1$
elements is an 
equiangular FUNTF, and that such ETFs are so-called group frames, see \cite{wald2018}.
\end{defn}

Amazingly, and elementary to prove, 
the finite {\it frames} for $\HH = {\mathbb K}^d$ are precisely the finite sequences, $X = \{x_h\}_{h=1}^N \subseteq
{\mathbb K}^d,$ that span ${\mathbb K}^d,$ i.e.,
\begin{equation}
\label{eq:span}
    \forall \,x \in {\mathbb \KK}^d, \; \exists\,c_1,\dots, c_N\in {\mathbb K}
    \quad  {\rm such \; that} \quad x = \sum_{h=1}^N\, c_h\,x_h.
\end{equation}
The innocent and Parseval-like Definition \ref{defn:frame} is the basis (sic) for the
power of frames, and it belies the
power of frames in dealing with numerical stability, robust signal representation,
and noise reduction problems,
see, e.g., \cite{daub1992}, \cite{BenFra1994} Chapters 3 and 7,
\cite{chri2016}, \cite{KovChe2007a}, and \cite{KovChe2007b}.

Let $X = \{x_h\}_{h \in J}$ be a frame for $\HH$. We define 
the following operators associated with every frame; they 
are crucial to frame theory. The {\it analysis operator} $L : \HH \rightarrow \ell^2(J)$ is defined by
\[
       \forall x \in  \HH,  \quad Lx = \{\inner{x}{x_h} \}_{h \in J}.
\]
The adjoint of the analysis operator is the {\it synthesis operator} 
$L^\ast : \ell^2(J) \rightarrow \HH$, and it is defined by
\[
        \forall a \in \ell^2(J),  \quad L^\ast a= \sum_{h \in J} a_h x_h.
\]
The {\it frame operator} is the mapping $\Fcal : \HH \rightarrow \HH$ defined as 
$\Fcal = L^\ast L$, i.e., 
\[
\forall x \in \HH, \quad \Fcal(x) = \sum_{h \in J} \inner{x}{x_h}  x_h.
\]

The following is a fundamental theorem.

\begin{thm}
\label{thm:framerecon}
Let $\HH$ be a separable Hilbert space, and let $X = \{x_h\}_{h \in J} \subseteq \HH$. 

{\it a.} $X$ is
 a frame for $\HH$ with frame bounds $A$ and $B$ if and only if 
 $\Fcal : \HH \rightarrow \HH$ is a topological isomorphism with norm bounds
 $\norm{\Fcal}_{op} \leq B$ and $\norm{\Fcal}_{op}^{-1} \leq A^{-1}$.  
 
 {\it b.} In the case of either condition of part {\it a}, we have the following:
 \begin{equation}\label{eq:frameopinverse}
B^{-1}I \leq \Fcal^{-1} \leq A^{-1} I,
\end{equation}
$\{ \Fcal^{-1} x_h \}$ is a frame for $\HH$ with frame bounds $B^{-1}$ and $A^{-1}$, and 
\begin{equation}
\label{eqn:recon}
        \forall x \in \HH, \quad x = \sum_{h \in J} \inner{x}{x_h}\Fcal^{-1} x_h = 
        \sum_{h \in J} \inner{x}{\Fcal^{-1}x_h} x_h
        = \sum_{h \in J} \inner{x}{\Fcal^{-1/2}x_h}\Fcal^{-1/2}x_h.
\end{equation}

\end{thm}

For a proof of part {\it a}, see \cite{BenWal1994}, pages 100--104.
For part {\it b},
let $X = \{x_h\}_{h \in J}$ be a frame for $\HH$. Then, the frame operator
$\Fcal$ is invertible (\cite{daub1992}, \cite{bene1994}); and $\Fcal$ is a multiple 
of the identity precisely when $X$ is a tight frame. 
Further, $\Fcal^{-1}$ is a positive self-adjoint operator and has a square root $\Fcal^{-1/2}$ 
(Theorem 12.33 in \cite{rudi1991}). This square root can be written as a power series in $\Fcal^{-1}$; 
consequently, it commutes with every operator 
that commutes with $\Fcal^{-1},$ and, in particular, with $\Fcal.$ These properties allow us to assert that 
$\{ \Fcal^{-1/2} \, x_h \}$ {\it is a Parseval frame for} $\HH$, and give the third equality of \eqref{eqn:recon},
see \cite{chri2016}, page 155.

The following is straightforward to prove, e.g., see \cite{CasKut2012}, \cite{wald2018}.

\begin{prop}
\label{prop:parsprops}
Given $\HH = \KK^d$ and $N \geq d$, and let $X = \{ x_j \}_{j=1}^N \subset \HH$.

{\it a.} If $X$ is a Parseval frame for $\HH$ and each $\norm{x_j} =1$, then $N=d$ and $X$ is an ONB for
$\HH$.

{\it b.} If $X$ is a FUNTF for $\HH$ and not an ONB for $\HH$, then the frame constant $A \neq 1$.

{\it c.} A FUNTF, resp., Parseval frame, for $\HH$ is not a Parseval frame, resp.,
FUNTF for $\HH$, unless $N=d$ and $X$ is an ONB for
$\HH$.

{\it d.} If $X$ is an equi-normed $A$-tight frame for $\HH$, then each  
$\norm{x_j} = \big(\frac{Ad}{N}\big)^{1/2}$.

{\it e.} If $X$ is a Parseval frame for $\HH$, then each $\norm{x_j} \leq1$.
The same result is true for any separable Hilbert space over $\KK$.

{\it f.} If $X$ is an equiangular, $A$-tight frame for $\HH$, then 
\[
   \forall j,k = 1, \ldots, N, \quad \norm{x_j} = \Big(\frac{Ad}{N}\Big)^{1/2} \;\; {\rm and} \;\;
   \big| \langle x_j, x_k \rangle \big| = \frac{A}{N} \sqrt {\frac{d(N-d)}{N - 1}}.
\]

{\it g.} If $X$ is an equiangular, Parseval frame for $\HH$, then each $\norm{x_j} <1$.

\end{prop}

\begin{rem}[Frames and bases for $\HH$]
In light of 
the fact that ONBs are 
frames, it is natural to ask to what extent frames can be constructed in terms of
ONBs. 
\begin{enumerate}

\item It may  be considered surprising that any infinite dimensional $\HH$ contains a frame for $\HH$
which does not contain a basis for $\HH$. The result is due to Casazza and Christensen, see \cite{chri2016}, 
Chapter 7, for details.
\item The first result relating frames and sums of bases is due to Casazza \cite{casa1998}.
{\it 
Let $\HH$ be a separable Hilbert space, and let $X =\{ x_h \}_{h \in J}$ be a frame for $\HH$
with upper frame bound $B$. Then, for every $\epsilon > 0$, there are ONBs $\{u_h\}_{h \in J},
\{v_h\}_{h \in J}, \{w_h\}_{h \in J}$ for $\HH$ and a constant $C =B(1+\epsilon)$ such that
}
\[
        \forall h \in J,\quad x_h = C(u_h + v_h + w_h).
\]
The proof depends on an operator-theoretic argument.
\end{enumerate}
\end{rem}

%%%%%%%%%%%%%%%%%%%%

\subsection{POVMs}
\label{sec:povm}
Definition \ref{defn:povmmeas} is a measure theoretic formulation of POVMs, see 
\cite{AliAntGaz1993}, \cite{BusGraLah1997} for applications to coherent states
and quantum physics, and see
\cite{BenCza2009} for the measure theory.

\begin{defn}[POVM]
\label{defn:povmmeas}
Let $\SS$ be a set, let $\Bcal$ be a $\sigma$-algebra of subsets of
$\SS$, and let $\HH$ be a separable Hilbert space. In this setting, a POVM 
on $\Bcal$ is a representation-like
mapping, $\mu: \Bcal \longrightarrow \Lcal(\HH)$, with the following properties:
\begin{enumerate}
\item $\forall \,U \in \Bcal, \, \mu(U) \in \Lcal(\HH)$  is a positive semi-definite self-adjoint operator;
\item $\mu(\emptyset) = 0$, the $0$-operator;
\item For every disjoint collection, $\{U_j \}_{j=1}^{\infty} \subseteq \Bcal $, if $x,y \in \HH$, then
\[
       \inner{\mu( \cup_{j=1}^{\infty} U_j)(x)}{y}  = \sum_{j=1}^{\infty} \inner{\mu(U_j)(x)}{y};
\]
\item  $\mu(\SS) = I$, the identity operator.
\end{enumerate}

\end{defn}

$\Lcal(\HH)$ is a non-commutative $\ast$-Banach algebra with unit, see \cite{AndBenDon2019}.

\begin{prop}
\label{povmmeas}
 Let $\{x_j\}_{j\in J}$ be a Parseval frame for $\HH$, where $\SS = J \subseteq \ZZ$.
Define a family $\{ \mu(U)\}_{U \subseteq J}$ of linear operators on $\HH$
by the formula,
\[
     \forall \, x \in \HH, \quad  \mu(U)(x)  =  \sum_{j \in U} \inner{x}{x_j} x_j.
\]
Then, $\mu$ is a POVM on $\Bcal$. If $\HH = \KK^d$, then we typically take $J = \{1, \ldots, N\}, \, N\geq d$.
\end{prop}

\begin{proof}
By direct manipulation with the definition of $\mu(U)$, we verify the first three
criteria of Definition \ref{defn:povmmeas}. The last criterion follows since 
$\{x_j\}_{j\in J}$ is a Parseval frame for $\HH$; in fact,
\[
       \forall \, x \in \HH, \quad  \mu(\SS)(x)  =  \sum_{j \in J} \inner{x}{x_j} x_j = x.
\]
\end{proof}

\begin{prop}
\label{parsmeas}
Let $\Bcal = \Pcal (\SS)$ be the power set $\sigma$-algebra of a countable set $\SS$,
and let
$\mu : \Bcal \longrightarrow \Lcal(\HH)$ be a POVM. 
Then, there is a countable set $J$, a Parseval frame
$\{ x_j\}_{j \in J}$ for $\HH$, and a disjoint partition $\{B_i\}_{i \in \SS}$ of $J$ such
that
\[
    \forall i \in \SS \;\; {\rm and} \;\; \forall x \in \HH, \quad  \mu(i)(x)  =  
    \sum_{j \in B_i} \inner{x}{x_j} x_j.
\]
Furthermore, if $\HH = {\KK}^d$, then each $B_i$ may be taken to be finite.
\end{prop}

\begin{proof}
For each $i \in \mathbb{S}$, $\mu(i) \in \mathcal{L}(\HH)$ is self-adjoint and positive 
semi-definite by definition.
To fix ideas, let $\HH = \KK^d$. (For more general $\HH$, there are appropriate
versions of the spectral theorem, that we now apply to $\KK^d$.)
By Theorem \ref{thm:spectral},
for each $i \in \mathbb{S}$, 
there is a $d$-element indexing set $B_i$, an orthonormal set 
$\{ v_j \}_{j \in B_i} \subseteq \KK^d$, and 
a set $\{ \lambda_j \}_{j \in B_i}$ of non-negative numbers  such that
\begin{equation*}
\forall x \in \KK^d, \;\; \mu(i)(x) = \sum_{j \in B_i} \lambda_j \inner{x}{v_j} v_j 
       = \sum_{j \in B_i} \inner{x}{x_j} x_j,
\end{equation*}
where
\begin{equation*}
\forall j \in B_i, \hspace{1em} x_j = \sqrt{\lambda_j}{v_j}.
\end{equation*}
Further, all the $B_i$ are disjoint.
Set $J = \cup_{i \in \mathbb{S}} B_i$. Because $\mathbb{S}$ is countable and each $B_i$ 
is finite, $J$ itself is countable. Since $\mu(\mathbb{S}) = I$, we have that
\begin{equation*}
\forall x \in \KK^d, \hspace{1em} x = \mu(\mathbb{S})(x) 
    = \sum_{i \in \mathbb{S}} \sum_{j \in B_i} \inner{x}{x_j} x_j = \sum_{j \in J} \inner{x}{x_j} x_j.
\end{equation*}
It follows that $\{ x_j \}_{j \in J}$ is a Parseval frame for $\KK^d$.
\end{proof}

Proposition \ref{parsmeas} is a converse of Proposition \ref{povmmeas} for the $\sigma$-algebra
of all subsets of $\SS = \ZZ$. Applicably, we can also say that
analyzing quantum measurements with a discrete set $\SS$ of outcomes 
is equivalent to analyzing
Parseval frames. Propositions \ref{povmmeas} and \ref{parsmeas} were established 
to illustrate  
the role of POVMs in quantum detection 
\cite{BenKeb2008} (2008), which itself depends on frame potential theory  \cite{BenFic2003} (2003).
Definition \ref{defn:povmmeas} is stated in some generality 
so as to be able eventually to extend the 
analysis of quantum measurements for more robust sets of outcomes. Also, in the 
setting of more general measurable spaces $\SS$ than $\ZZ$, there are corresponding
equivalences with Parseval frames. In fact, each of the propositions in this 
subsection has a significant, straightforward 
generalization, which we do not pursue herein.

\begin{example} [Resolution of the identity]
\label{ex:resid}
{\it a.} Given $\SS$, $\Bcal$, and $\HH$ as in Definition \ref{defn:povmmeas}. Take $\HH$ over $\CC$.
A {\it resolution of the identity} on $\Bcal$ is a mapping, 
$\rho: \Bcal \longrightarrow \Lcal(\HH)$, with the following properties:
$\rho(\emptyset) = 0$, $\rho(\SS) = I$; each $\rho(U)$ is a self-adjoint
projection and so each $\rho(U)$ is positive semi-definite ($\langle \rho(U)(x), x \rangle=
\norm{\rho(U)(x)}^2$ for all $x \in \HH$); $\rho(U \cap V) = \rho(U)\rho(V)$ (composition) on 
$\Bcal$; $\rho$ is finitely additive on $\Bcal$; and $\rho_{x,y}: \Bcal \to \CC$ defined by
\[
   \forall x, y \in \HH, \quad \rho_{x,y}(U) = \langle \rho(U)(x), y \rangle
\]
is a complex measure on $\Bcal$. The importance/existence of resolutions of
the identity is the spectral theorem that asserts that every bounded 
self-adjoint (and more generally) operator $A$ on $\HH$ induces a resolution of the identity $\rho$,
and $A$ can be reconstructed from $\rho$ in terms of a certain type of integral.

{\it b.} With the set-up of part {\it a}, suppose (the weak hypothesis) that $\SS$ can
be written as a disjoint union $\cup U_n$ of a sequence $\{ U_n \} \subseteq \Bcal$.
Define $E_n = \mu(U_n) : \HH \to \HH$, where $\mu$ is given in
Definition \ref{defn:povmmeas}. Then, we have the resolution of the identity
$I = \sum E_n$, because
\[
    \forall x, y \in \HH, \; \langle I(x), y \rangle = \langle \mu(\SS)(x), y \rangle
    =\langle \mu (\cup U_n)(x), y \rangle = \langle \sum \mu(U_n)(x), y  \rangle
    = \langle \sum E_n (x), y \rangle.
\]

{\it c.} We mention parts {\it a} and {\it b} since we shall be dealing with special POVMs
described in Definition \ref{defn:povmsaop}. These will correspond to discrete
observables in quantum measurement, and the domain
$\Bcal$ does not play an explicit role.

\end{example}

\begin{defn}[POVMs, effects, and projections]
\label{defn:povmsaop}
 Let $\Ucal(\HH)$ denote the set of operators $U \in \Lcal_+(\HH)$
 for which $0 \leq U \leq I$; and let $\mathcal{E}(\HH)$ denote the set of 
 operators $E \in \Scal_+(\HH)$ for which $0 \leq E \leq I$,
 i.e., $\mathcal{E}(\HH) = \Scal_+(\HH) \cap \Ucal(\HH)$.
 Note that if $A \in \Lcal_+(\HH) \setminus \Ucal(\HH)$, then there is $a > 1$ and $U \in \Ucal(\HH)$
 such that $A = aU$.
 In fact, set $a = \norm{A} > 1$ and $U = (1/ \norm{A}) A$. The verification is immediate, e.g.,
 $ U \leq I$ since $\langle (I-U)(x),x \rangle = \langle x,x \rangle - (1/ \norm{A}) \langle A(x),x \rangle$
 and $(1/ \norm{A}) \langle A(x),x \rangle \leq \norm{x}^2 = \langle x,x \rangle$.
 Similarly, if $A \in \Scal_+(\HH) \setminus \Ecal(\HH)$, then there is $a > 1$ and $E \in \Ecal(\HH)$
 such that $A = aE$.
  $\mathcal{E}(\HH)$
is the set of all {\it effects}.

A {\it positive operator-valued measurement} on $\HH$, that we also designate
 by POVM, is a 
sequence $\{ E_n \} \subseteq \mathcal{E}(\HH)$ such that $I = \sum E_n$,
see \cite{BusGraLah1997}, \cite{busc2003},
\cite{RenBluScoCav2004}. 
For example, if $E \in \Ecal(\HH)$, then $I - E \in \Ecal(\HH)$ 
since $I - E \in \Scal_+(\HH)$ and $I - E \leq I$, 
in particular, $0 \leq I - E \leq I$; and thus $\{ E, I-E \}$ is a 
POVM on $\HH$.

 Let $\Pcal_+(\HH) \subseteq \Lcal(\HH)$ be the space of {\it self-adjoint
projections}. 
If $P \in \Pcal_+(\HH)$, then $P \geq 0$ since $\langle P(x),x \rangle = \langle P^2(x), x \rangle
= \norm{P(x)}^2$. Further, $P \in \Pcal_+(\HH)$ implies $I - P \in \Pcal_+(\HH)$, and so we have
that $\Pcal_+(\HH) \subseteq \Ecal(\HH)$.

POVMs, effects, and projections are the topic of Section \ref{sec:appl}.

\end{defn}

\begin{defn}[Tensor product and ket-bra notation]
\label{defn:tensorket}
{\it a.} For given $\mathbb{H}$ over $\KK$, let $\mathbb{H}^\prime$ 
be the 
dual space $\Lcal(\HH, \KK)$ of bounded linear functionals $L: \mathbb{H} \to \KK$,
taken with the operator norm topology given by $\norm{L} = {\rm sup}_{\norm{x} \leq 1} |L(x)|$. 
A fundamental result, which is essentially the Riesz representation theorem for Hilbert spaces,
is that {\it there is a conjugate-linear surjective isometry}, $\HH \to \HH^{\prime}$,
$y \mapsto L_y = y^*$ {\it defined by the formula}
\[
    \forall x \in \HH, \quad L_y(x) = y^*(x) = \langle x, y \rangle,
\]
{\it where conjugate-linear means that} $\langle x, a_1 y_1 + a_2 y_2 \rangle =
\overline{a_1} \langle x,y_1\rangle  + \overline{a_2} \langle x, y_2 \rangle$.

{\it b.} The {\it tensor product} $\otimes: \mathbb{H} \times \mathbb{H}^\prime \to \mathcal{L}(\mathbb{H})$ 
is the bilinear mapping sending pairs $(x, y^*)$ to linear operators $x \otimes y^*$ 
defined by the action,
\begin{equation*}
       \forall z \in \mathbb{H}, \quad (x \otimes y^*)(z) = (y^*(z)) x = \inner{z}{y} x.
\end{equation*}
In this definition, we note that if $\KK = \CC$, then we do not generally have $ay^* = (ay)^*$, and the
bilinearity follows since $(x \otimes (ay^*))(z) = \langle z, \overline{a} y \rangle x = a (x \otimes y^*)(z)$.

{\it c.} Let $x = (x_1, \ldots, x_d)$, $y^* = (y_1, \ldots, y_d) \in \HH =\KK^d$. The {\it outer product} 
$x y^*$ is the $d \times d$
matrix $(z_{i,j})$, where $z_{i,j} = x_i \overline{y_j}$; and, in fact, this is the 
{\it tensor product} $x \otimes y^*$ defined more generally in part {\it b}.  $x \otimes y^*$ is clearly a rank-$1$ operator
on $\KK^d$ since each of the columns of $(z_{i,j})$ is a multiple of the
first column. In Dirac notation, $x \otimes y^*$ is the {\it ket-bra} $|x \rangle\langle y |$.
\end{defn}

\begin{lem}
\label{lem:sapossd}
Given $\HH$ and let $x \in \HH$. Define $E = x \otimes x^* \in \Lcal(\HH)$.
Then, $E \in \Scal_{+}(\HH)$, i.e., $E$ is self-adjoint and positive semi-definite.
\end{lem}

\begin{proof}
For any any $y, z \in \mathbb{H}$ we have the equations,
\begin{equation*}
\inner{E (y)}{z} = \inner{(x \otimes x^*)(y)}{z} = \inner{(x^*(y))x}{z} 
= \inner{y}{x} \inner{x}{z} 
\end{equation*}
and
\begin{equation*}
\inner{y}{E(z)} = \inner{y}{(x \otimes x^*)(z)} = \inner{y}{(x^*(z))x} 
= \overline{\inner{z}{x}} \inner{y}{x} = \inner{x}{z} \inner{y}{x}.
\end{equation*}
Therefore,
\begin{equation*}
\inner{E (y)}{z} = \inner{y}{x} \inner{x}{z} = \inner{x}{z} \inner{y}{x} = \inner{y}{E (z)}
\end{equation*}
and
\begin{equation}
\label{eq:possd}
\inner{E(y)}{y} = \inner{y}{x} \inner{x}{y} = \inner{y}{x} \overline{\inner{y}{x}} = |\inner{y}{x}|^2 \geq 0.
\end{equation}
Consequently, $E$ is self-adjoint and positive semi-definite.
\end{proof}

\begin{prop}
\label{prop:povm}
 Let $\{x_j\}_{j = 1}^N$ be a Parseval frame for $\HH = \KK^d$. Then, 
 $\{E_j :=x_j \otimes x_j^*\}_{j = 1}^N \subseteq \Ecal(\KK^d)$ 
is a POVM on $\HH$.
\end{prop}

\begin{proof}
First note that each $E_j \in  \Ecal(\HH)$. In fact, $E_j \leq I$ by \eqref{eq:possd} since
\[
       \inner{E_j(y)}{y} = |\inner{y}{x_j}|^2 \leq \norm{y}^2\norm{x_j}^2 \leq \langle y,y \rangle,
\]
where we have used the fact that Parseval frames are contained in 
the closed unit ball of $\HH$ (stated for $\HH = \KK^d$
in Proposition \ref{prop:parsprops}). This is also a consequence of the resolution of the
identity formula that we shall now verify, since the $E_j$ are positive semi-definite.
 For any $y \in \mathbb{H}$, we have
\begin{equation*}
\inner{\sum_{j =1}^N E_j (y)}{y} = \sum_{j =1}^N \inner{E_j y}{y} = \sum_{j =1}^N |\inner{y}{x_j}|^2 = \norm{y}^2
\end{equation*}
by the Parseval condition, so that any eigenvalue of $\sum_{j = 1}^N E_j$ must have 
absolute value $1$. Each $E_j$ is self-adjoint and positive semi-definite (Lemma \ref{lem:sapossd}),
 and hence $\sum_{j = 1}^N E_j$ 
is self-adjoint and positive semi-definite. Thus, each eigenvalue of $\sum_{j = 1}^N E_j$ 
must be real and non-negative. Combining these facts shows that $1$ is the only 
eigenvalue of the operator $\sum_{j = 1}^N E_j$. The spectral theorem then implies that $\sum_{j =1}^N E_j$ 
is the identity operator.
\end{proof}

Conversely, given any POVM $\{E_j\}_{j \in J}$ on $\KK^d$, we can construct a Parseval
frame $\{x_j\}_{j = 1}^N$ for $\KK^d$ from 
the eigenvectors of the $E_j$ in the following way. The hypothesis that we are given a POVM
is only used in the penultimate equality of the following proof.

\begin{prop}
\label{prop:pars}
Let $\HH = \KK^d$, and let $\{ E_j \}_{j \in J} \subseteq \Ecal({\HH})$ be a POVM on $\mathbb{H}$. 
There exists a Parseval frame 
$\{ x_{j,k} \}_{j \in J, 1 \leq k \leq d}$ for $\mathbb{H}$ such that for each $E_j$ we have 
$E_j = \sum_{k=1}^d x_{j,k} \otimes x_{j,k}^*$.
\end{prop}

\begin{proof}
For each $E_j$ we invoke the spectral theorem to choose an eigenbasis 
$\{ e_{j,k} \}_{k=1}^d$ for $E_j$ corresponding to the (real, non-negative, not necessarily distinct) 
eigenvalues $\lambda_{j,k}, \, k=1, \ldots, d$ of $E_j$. Then, for each $j \in J$ 
and $k \in \{1, \dots, d\}$ set $x_{j,k} = \sqrt{\lambda_{j,k}} e_{j,k}$. For any $y \in \mathbb{H}$ 
and any $j \in J$, we can write $y = \sum_{k=1}^d y_k e_{j,k}$. Next, we compute
\begin{align*}
(\sum_{k=1}^d x_{j,k} \otimes x_{j,k}^*) y = \sum_{k=1}^d \inner{y}{x_{j,k}} x_{j,k} 
&= \sum_{k=1}^d \inner{y}{e_{j,k}} \lambda_{j,k} e_{j,k} \\
&= \sum_{k=1}^d y_k \lambda_{j,k} e_{j,k} \\
&= \sum_{k=1}^d y_k E_j (e_{j,k}) = E_j (\sum_{k=1}^d y_k e_{j,k}) = E_j(y),
\end{align*}
so that $E_j = \sum_{k=1}^d x_{j,k} \otimes x_{j,k}^*$. Furthermore,
\begin{align*}
\sum_{j \in J} \sum_{k=1}^d |\inner{x_{j,k}}{y}|^2 &= \sum_{j \in J} \sum_{k=1}^d \inner{y}{x_{j,k}} \inner{x_{j,k}}{y} \\
&= \sum_{j \in J} \sum_{k=1}^d \inner{\inner{y}{x_{j,k}} x_{j,k}}{y} \\
&= \sum_{j \in J} \sum_{k=1}^d \inner{(x_{j,k} \otimes x_{j,k}^*)(y)}{y} \\
&= \inner{(\sum_{j \in J} \sum_{k=1}^d x_{j,k} \otimes x_{j,k}^*)y}{y} 
= \inner{(\sum_{j \in J} E_j)(y)}{y} = \inner{y}{y} = \norm{y}^2,
\end{align*}
where the penultimate equality follows since 
$\{ E_j \}_{j \in J} \subseteq \Ecal({\HH})$ is a POVM on $\HH$.
Therefore, $\{ x_{j,k} \}_{j \in J, 1 \leq k \leq d}$ is a Parseval frame for $\HH = \KK^d$.
\end{proof}

\section{Gleason functions for Parseval frames}
\label{sec:gleapars}

\subsection{Quadratic forms are Gleason functions for Parseval frames}
\label{sec:props}
Suppose $f: \KK^d \longrightarrow \KK$ is
a function for which there exists $W \in \KK$ such that for any frame $\{x_j\}_{j\in J}$ for $\KK^d$
we have $\sum_{j\in J} f(x_j) =W$. Such a function must be 
identically zero since one can add an arbitrary vector
to any frame and the set remains a frame. Therefore,
a more specific class of frames must be examined in order to extend Gleason's
theorem to frames. The clear choices are 
Parseval frames or FUNTFs,
since Parseval frames and FUNTFs reduce to 
ONBs when the cardinality of the frame is the dimension of the 
Hilbert space, see Proposition \ref{prop:parsprops}. Given that Gleason was originally concerned with 
measures corresponding to quantum measurement, and since Parseval frames directly correspond
to positive operator-valued measures (Subsection \ref{sec:povm}), which in turn are a 
general form of quantum measurement,
we shall extend the notion of Gleason's functions 
to Parseval frames as promised in Subsection 
\ref{sec:gleagoal}.

The spectral theorem and a straightforward calculation give Theorem \ref{thm:gleapars}, which
is a direct generalization of Theorem \ref{thm:saimpliesglea}.

\begin{thm}
\label{thm:gleapars}
Let $\HH = \KK^d$, and let
$A$ be a self-adjoint linear operator $A: \HH \longrightarrow \HH$.
The function $g: B^d \longrightarrow \KK$, defined by
\[
     \forall x\in B^d, \quad  g(x)=\langle A(x),x\rangle,
 \]
is a Gleason function 
of weight $W = {\rm tr}(A)$ for the finite Parseval frames $X$ for $\HH$.
Clearly, $g \in L^{\infty}(B^d)$ with $\norm{g}_{L^{\infty}(B^d)} \leq \norm{A}_{op}$, $g(0) = 0$, and
\[
   \forall x \in B^d \; {\rm and} \; \forall \alpha \in \KK, \; {\rm where} \; |\alpha| \leq 1, \quad 
   \alpha x \in B^d \; {\rm and} \; g(\alpha x) = |\alpha|^2 g(x).
\]
\end{thm}

\begin{proof} 
By the spectral
theorem, there exists an orthonormal eigenbasis $\{e_j\}_{j=1}^d$
associated with the set $\{ \lambda_j \}_{j=1}^d$ of eigenvalues 
for $A$. Hence, for all $x \in \HH$, we have $x=\sum_{j=1}^d \inner{x}{e_j}e_j$ and 
$A(x)=\sum_{j=1}^d \inner{x}{e_j}\lambda_j e_j.$
If $X = \{x_j\}_{j=1}^N$ is a
Parseval frame for $\HH$, then we compute
\[
        {\rm tr}(A)=\sum_{j=1}^{d}\lambda_j
=\sum_{j=1}^d \lambda_j\norm{e_j}^2
=\sum_{j=1}^d \lambda_j\sum_{n=1}^N |\inner{x_n}{e_j}|^2,
\]
where the last equality follows since the frame
is Parseval. 
Re-ordering the finite sums yields the desired result:
\[
       {\rm  tr}(A)= \sum_{n=1}^N\sum_{j=1}^d \lambda_j  \inner{e_j}{x_n} \inner{x_n}{e_j}
 =\sum_{n=1}^N\sum_{j=1}^d \inner{\inner{x_n}{e_j} \lambda_j e_j}{x_n}
  =\sum_{n=1}^N\inner{A(x_n)}{x_n}=\sum_{n=1}^N g(x_n).
\]
Therefore, $g$ is a Gleason function of weight $W = {\rm tr}(A)$ for the finite Parseval
frames $X$
 for $\HH$.
\end{proof}

\begin{rem}[Generalizations of Theorem \ref{thm:gleapars}]
\label{rem:normal}
{\it a.} Theorem \ref{thm:gleapars} is true for normal operators over $\CC$, 
since the spectral theorem is true in that setting, e.g., see
\cite{FriInsSpe1997}, page 377.
Further, Theorem \ref{thm:gleapars} is true for arbitrary linear operators $A$ over $\HH = \CC^d$, 
since every such $A$ can be written as $A = B + iC$, 
where $B$ and $C$ are self-adjoint linear operators, see, e.g.,
\cite{halm1958}, Section 70.

{\it b.} The proof of Theorem \ref{thm:gleapars} can be extended to infinite Parseval frames 
$X = \{x_n\}_{n=1}^{\infty}$ for $\KK^d$ by noting that the series
$\sum_{j=1}^d \sum_{n=1}^{\infty} \lambda_j  |\inner{x_n}{e_j}|^2 < \infty$
is absolutely convergent so that the terms can be rearranged,
see Remark \ref{rem:trace}.

{\it c.} See the Problem stated in Subsection \ref{sec:thmprob} for a role that the condition
$g(0) = 0$ in Theorem \ref{thm:gleapars} plays.

\end{rem}

In the remainder of this Section \ref{sec:gleapars}, we shall establish broad
conditions which imply that a Gleason function $g$ for all Parseval frames for $\HH = \KK^d$
is of the form $\langle A(x),x\rangle$ for some linear operator $A$. We shall first focus on
the case in which $g$ is continuous or non-negative, and then extend these results to bounded 
Gleason functions.

%%%%%%%%%%%%%%%%%%%%%%%%%%%%

\subsection{Basic properties}
\label{sec:basicprops}

This subsection collects some general facts about Gleason functions 
for the finite Parseval frames for $\HH = \KK^d$, that we shall use in the sequel.

The first two results are elementary and are stated on their own to avoid unnecessary repetition.
\begin{prop}
\label{prop:prop1}
Let $\HH = \KK^d$. Then, the Gleason functions for the finite Parseval frames for $\HH$ 
form a $\KK$-vector space under pointwise addition of functions and scalar multiplication.
\end{prop}
\begin{proof}
Suppose $f, g: B^d \to \KK$ are two Gleason functions of weights $W_1, W_2$, 
respectively, for the Parseval frames for $\HH$, and let $\alpha, \beta \in \KK$. 
We show that $\alpha f + \beta g$ is a Gleason function of weight 
$\alpha W_1 + \beta W_2$ for the Parseval frames for $\HH$. Let $\{ x_i \}_{i=1}^N \subseteq B^d$
be a Parseval 
frame for $\HH$. Then,
\[
\sum_{i=1}^N (\alpha f + \beta g)(x_i) = \sum_{i=1}^N \alpha f(x_i) + \beta g(x_i) = \alpha \sum_{i=1}^N f(x_i) 
+ \beta \sum_{i=1}^N g(x_i) = \alpha W_1 + \beta W_2,
\]
as claimed.
\end{proof}

\begin{prop}
\label{prop:prop2}
Let $\HH = \KK^d$ and let $g$ be a Gleason function of weight $W$ for the finite Parseval 
frames for $\HH$. Then $g(0) = 0$.
\end{prop}
\begin{proof}
Let $\{ v_1, \dots, v_d \}$ be any ONB for $\HH$. Then, as a consequence of the 
Parseval identity for ONBs,
both $\{ v_1, \dots, v_d \}$ and $\{ 0, v_1, \dots, v_d \}$ are Parseval frames for $\HH$. Because 
$g$ is a Gleason function of weight $W$ for the Parseval frames for $\HH$, we have
\[
      \sum_{i=1}^d g(v_i) = W = g(0) + \sum_{i=1}^d g(v_i).
\]
Thus, $g(0) = 0$.
\end{proof}

The following is a key lemma in obtaining information on the values of Gleason functions for Parseval frames.

\begin{lem} 
\label{lem:alphai}
Let $\HH = \KK^d$ and let $g$ be a Gleason function of weight $W$ for the finite Parseval frames 
for $\HH$. Let $\{ \alpha_1, \dots, \alpha_n \}$ be a finite sequence in $\KK$ such that 
$\sum_{i=1}^n |{\alpha_i}|^2 = 1$. Then,
\[
\forall x \in B^d, \hspace{2em} \sum_{i=1}^n g(\alpha_i x) = g(x).
\]
\end{lem}
\begin{proof}
If $x = 0$, then $\alpha_i x = \alpha_i \cdot 0 = 0$ for all $i$; and,
using Proposition \ref{prop:prop2}, it follows that
\[
\sum_{i=1}^n g(\alpha_i x) = 0 = g(0),
\]
as claimed.

Otherwise $x \in B^d \setminus \{0\}$, and set $x_1 = \beta x$ where 
$\beta = (1 - \norm{x}^2)^{1/2} / \norm{x}$. 
Choose an ONB $X_1 := \{ u_1, \dots, u_{d-1} \}$ for the orthogonal complement 
$Y^\bot$ of the closed linear span $Y$ of $x$. Then, the sequences 
$X := \{ x \} \cup \{ x_1 \} \cup X_1$ and 
$X' := \{ \alpha_1 x, \dots, \alpha_n x \} \cup \{ x_1 \} \cup X_1$ are both Parseval
frames for $\HH$. To verify this claim, begin by taking any
$y \in \HH$. Let $v = x / \norm{x}$, so that $\{ v, u_1, \dots, u_{d-1} \}$ is an 
ONB for $\HH$. Then, by Parseval's identity for ONBs,
\begin{align*}
    \sum_{u \in X} |\inner{y}{u}|^2 
    &= |\inner{y}{x}|^2 + 
    \frac{1 - \norm{x}^2}{\norm{x}^2} |\inner{y}{x}|^2 + 
    \sum_{u \in X_1} |\inner{y}{u}|^2 + |\inner{y}{v}|^2 - |\inner{y}{v}|^2 \\
    &= |\inner{y}{x}|^2 + \frac{1 - \norm{x}^2}{\norm{x}^2} |\inner{y}{x}|^2 + 
    \norm{y}^2 - \frac{1}{\norm{x}^2} |\inner{y}{x}|^2 = \norm{y}^2.
\end{align*}
Hence, $X$ is a Parseval frame for $\HH$. Because
\[
\sum_{i=1}^n |\inner{y}{\alpha_i x}|^2 = \sum_{i=1}^n |\alpha_i|^2 |\inner{y}{x}|^2
 = |\inner{y}{x}|^2,
\]
the above calculations imply that $X'$ is also a Parseval frame for $\HH$.

Because $X$ and $X'$ are both Parseval frames for $\HH$ and $g$ is a Gleason function of weight 
$W$ for the finite Parseval frames for $\HH$, we have
\[
g(x) + g(x_1) + \sum_{u \in X_1} g(u) = W = 
\sum_{i=1}^n g(\alpha_i x) + g(x_1) + \sum_{u \in X_1} g(u).
\]
Canceling like terms gives the desired result.
\end{proof}

\begin{lem} 
\label{lem:rationalapprox}
Let $\HH = \KK^d$ and let $g$ be a Gleason function of weight $W$ for the finite Parseval 
frames for $\HH$. Let $x \in B^d$, and let $q \in \QQ$ be 
a non-negative rational number with the property that $\sqrt{q} x \in B^d$. Then, $g(\sqrt{q}x) = q\cdot g(x)$.
\end{lem}

\begin{proof}
First observe that for any $z \in B^d$ and for any positive integer $P$, the sequence 
$\{ \alpha_1, \dots, \alpha_P \}$ defined by $\alpha_i := 1 / \sqrt{P}$ for all $i$ satisfies the 
hypotheses of Lemma \ref{lem:alphai}. Hence,
\begin{equation}
   \label{eq:Zscaling}
    \forall z \in B^d \textrm{ and } \forall P \in \NN, \quad g(z) = P g\left(\frac{z}{\sqrt{P}}\right).
\end{equation}

Now let $x \in B^d$, and suppose first that $q = M / N \in \QQ \cap [0, 1]$, where $M, N$ 
are nonnegative integers. Clearly $N \neq 0$, and if $M = 0$ the proof 
that $g(\sqrt{q} x) = q \cdot g(x)$ is immediate since $g(0) = 0$. 
Thus, we may assume $M \neq 0$. By \eqref{eq:Zscaling}, $g(x) = N g(x / \sqrt{N})$. 
For 
$y = \sqrt{\frac{M}{N}} x$ we have $y \in B^d$
because $\norm{y} \leq \norm{x}$, and \eqref{eq:Zscaling} gives $g(y) = M g(y / \sqrt{M})$. 
Hence,
\[
\frac{M}{N} g(x) = M g\left(\frac{x}{\sqrt{N}}\right) = M g\left(\frac{y}{\sqrt{M}}\right) = g(y) = g\left(\sqrt{\frac{M}{N}} x\right).
\]

Otherwise, $q > 1$. Then, $1 / \sqrt{q} \in \QQ \cap [0, 1]$, and the results 
of the preceding paragraph imply that
\[
g(x) = g\left(\frac{1}{\sqrt{q}}  \sqrt{q}x\right) = \frac{1}{q} g(\sqrt{q} x).
\]
Multiplying both sides by $q$ yields the claim.
\end{proof}

%%%%%%%%%%%%%%%%%%%%%%%%%%%%%

\subsection{Gleason functions and quadratic forms}
\label{sec:gleaquad}

We showed in Theorem \ref{thm:gleapars} that 
a self-adjoint operator generates a Gleason function 
for Parseval frames that is a quadratic form. 
%By definition, quadratic forms 
%over $\HH = \KK^d$ are homogeneous functions $Q$
%of degree $2$; and, therefore, if $\alpha \in \KK$, then $Q(\alpha x)=|\alpha|^2 Q(x)$. 
%Following 's calculation \cite{busc2003}, where he showed that functions
%that are constant when
%summed over elements of POVMs are homogeneous of degree 1, we shall show that
%all Gleason functions for the Parseval frames for $\KK^d$ are homogeneous functions of degree 
%$2$ for $|\alpha| \in [0,1]$.
We have the following result, cf. Busch \cite{busc2003} and Caves et al. \cite{CavFucManRen2004}.

\begin{thm} 
\label{thm:continuoushomog}
Let $\HH = \KK^d$ and let $g$ be a Gleason function of weight $W$ for the finite Parseval frames 
for $\HH$. Suppose $g$ is continuous on $B^d$. Then, 
\[
       \forall x \in B^d \textrm{ and } \forall \alpha \in \KK \textrm{ with } |\alpha| \leq 1, \quad 
      g(\alpha x) = |\alpha|^2 g(x).
\]
\end{thm}
\begin{proof}
Let $x \in B^d$, and let $\alpha \in \KK$ satisfy $|\alpha| \leq 1$. If $\alpha = 0$ then we need to show 
$g(0) = 0$; but this has already been shown. Thus we may assume $\alpha \neq 0$. 
Let $\zeta = |\alpha| / \alpha$; then $|\zeta| = 1$. By Lemma \ref{lem:alphai}, 
$g(\alpha x) = g(\zeta \alpha x) = g(|\alpha| x)$. 

Hence, without loss of generality, we may take $\alpha \in (0, 1]$. Let $\{ q_n \}_{n=1}^\infty$ be a 
sequence in $\QQ \cap [0, 1]$ with $q_n \to \alpha$. Then $q_n^2 \to \alpha^2$, so by continuity 
of $g$ and Lemma \ref{lem:rationalapprox},
\[
g(\alpha x) = \lim_{n \to \infty} g(q_n x) = \lim_{n \to \infty} q_n^2 g(x) = \alpha^2 g(x)
\]
as claimed.
\end{proof}

\begin{thm} 
\label{thm:nonnegativehomog}
Let $\HH = \KK^d$ and let $g$ be a Gleason function of weight $W$ for the finite Parseval 
frames for $\HH$. Suppose $g$ is non-negative. Then, 
\[
\forall x \in B^d \textrm{ and } \forall \alpha \in \KK \textrm{ with } |\alpha| \leq 1, \quad g(\alpha x) 
= |\alpha|^2 g(x).
\]
\end{thm}
\begin{proof}
Let $x \in B^d$ and $\alpha \in \KK$ with $|\alpha| \leq 1$. As in the proof of 
Theorem \ref{thm:continuoushomog} we may assume $\alpha \in (0, 1]$. Also we may take 
$x \neq 0$, as if $x = 0$ then the claim is $g(0) = 0$, which has already been shown.

Let $\theta \in [0, \pi / 2]$. Then, by the Pythagorean theorem, the sequence $\{ \cos \theta, \sin \theta \}$ 
satisfies the hypotheses of Lemma \ref{lem:alphai}, so that $g(x) = g(\cos(\theta)x) + g(\sin(\theta)x)$. 

Consider the line segment $L_x := \{ \beta x : \beta \in [0, 1] \}$ extending from the origin to $x$. 
Let $0 \leq \gamma < \beta \leq 1$, and set $\theta = \cos^{-1}(\gamma / \beta) \in [0, \pi / 2]$. Set $y := \beta x$ and $z := \gamma x$, so that $y, z \in L_x$. Then
\[
g(y) - g(z) = g(y) - g(\gamma x) = g(y) - g(\cos(\theta) y) = g(\sin(\theta) y) \geq 0.
\]
Therefore, \emph{$g$ is monotonically increasing from $g(0) = 0$ on $L_x$}.

Now let $\{ p_n \}_{n=1}^\infty$ and $\{ q_n \}_{n=1}^\infty$ be sequences in $\QQ \cap [0, 1]$ with $\{ p_n \}$ decreasing, $\{ q_n \}$ increasing, and $\lim_{n \to \infty} p_n = \lim_{n \to \infty} q_n = \alpha^2$. Then $p_n g(x) \to \alpha^2 g(x)$ and $q_n g(x) \to \alpha^2 g(x)$. Also, by monotonicity of $g$ and Lemma \ref{lem:rationalapprox},
\[
q_n g(x) = g(\sqrt{q_n} x) \leq g(\alpha x) \leq g(\sqrt{p_n} x) = p_n g(x).
\]
Combining these claims gives the desired result. This is a standard technique, see
Busch \cite{busc2003} and Caves et al. \cite{CavFucManRen2004}.

\end{proof}

\begin{rem}[Continuity on rays]
 Let $g$ be a non-negative Gleason function for the Parseval frames for $\KK^d$.
 In Theorem \ref{thm:nonnegativehomog} we proved that along
any ray beginning at the origin, $g$ is an {\it increasing} function
beginning at $g(0) = 0$ and going out to the boundary of $B^d$.
%By the definition of continuity and $\eqref{eq:homog}$, an
%elementary calculation also allows us to assert further that $g$ {\it is continuous 
%on each such ray}.
\end{rem}

It is immediate from the definition
that, if $g$ is a Gleason function for the Parseval frames for $\KK^d$, then $g$ is a 
Gleason function for the ONBs for $\KK^d$. We have also noted that if $g$ is
defined by a self-adjoint linear operator, 
then $g$ is a Gleason function
for the Parseval frames for $\KK^d$ (Theorem \ref{thm:gleapars}). 
Using Gleason's original theorem, we shall now prove
various partial converses. 
%in Theorem \ref{thm:gleaparsiffsa}. This will be done by combining Gleason's original 
%theorem with Theorem \ref{thm:homog} and 
%the special case Theorem \ref{thm:naimark} of Naimark's dilation theorem (1940). 

We shall need the following result asserting that orthogonal projections
of ONBs are Parseval frames. It is an elementary converse to Naimark's theorem.

\begin{prop} 
\label{prop:frameprojection}
Let $\HH = \KK^d$ and let $\GG$ be a closed subspace of $\HH$. Write $P$ for the orthogonal 
projection of $\HH$ onto $\GG$. Let $\{ x_j \}_{j=1}^N$ be a Parseval frame for $\HH$. Then, 
$\{ P (x_j )\}_{j=1}^N$ is a Parseval frame for $\GG$. In particular, if $\{ x_j \}_{j=1}^N$ is an 
ONB for $\HH$, 
then $\{ P (x_j) \}_{j=1}^N$ is a Parseval frame for $\GG$.
\end{prop}
\begin{proof}
Let $y \in \GG$, so that $y = P(y)$. We verify the Parseval condition for $\{ P (x_j )\}_{j=1}^N$ by
direct calculation:
\[
\norm{y}^2 = \sum_{j=1}^N |\inner{y}{x_j}|^2 = 
\sum_{j=1}^N {\inner{P(y)}{x_j}}^2 = \sum_{j=1}^N |\inner{y}{P(x_j)}|^2
\]
where the last equality holds because orthogonal projections are normal.
\end{proof}

\begin{rem}[Naimark's theorem]
\label{rem:naim}
For Naimark's theorem generally, see,
e.g., Naimark \cite{naim1940}, \cite{naim1943}, Paulsen \cite{paul2003} (2003) 
in terms of POVMs, and Czaja \cite{czaj2008} (2008), cf. 
Chandler Davis \cite{davi1977} (1977). A beautiful idea dealing with their
dilation viewpoint on frames gave rise to
Han and Larson's theorem, see \cite{HanLar2000} (2000), Proposition 1.1. 
It is at once a special case of Naimark's theorem, it has an elementary proof different
from Naimark's formulation, it
generalizes significantly in terms of group representations, and it has broad
applicability, e.g., \cite{CasKov2003}, \cite{CasRedTre2008}.

\end{rem}

\begin{thm} 
\label{thm:nonnegativemainthm}
Let $\HH = \KK^d$. A non-negative function $g: B^d \to \RR$ is a Gleason function of weight 
$W$ for the finite Parseval frames for $\HH$ if and only if there exists a self-adjoint 
and positive semi-definite operator $A: \HH \to \HH$ with trace ${\rm tr} (A )= W$ such that
\[
       \forall x \in B^d, \quad g(x) = \inner{A(x)}{x}.
\]
\end{thm}
\begin{proof}
If such an $A$ exists then $g$ is a Gleason function of weight $W$ for the finite Parseval frames 
for $\KK^d$, and, in fact, for all Parseval frames for $\HH$, by Theorem \ref{thm:gleapars}.

For the converse, let $g$ be a Gleason function of weight $W$ for the finite Parseval frames 
for $\HH$. 
Let $N = \max(d, 3)$ and consider the Hilbert space $\KK^N$. We can naturally identify $\HH$ 
with the closed subspace of $\KK^N$ spanned by the first $d$ standard basis vectors. 
Let $P: \KK^N \to \HH$ be the projection onto the first $d$ coordinates, and define 
$F: B^N \to \RR$ by $F(x) := g(P(x))$ for all $x \in B^N$. 
For an arbitrary ONB $\{ e_i \}_{i=1}^N$ for $\KK^N$, $\{ P (e_i )\}_{i=1}^N$ is a Parseval 
frame for $\HH$ by Proposition \ref{prop:frameprojection}. Then, because $g$ is a 
Gleason function of weight $W$ for the finite Parseval frames for $\HH$, we have
\[
       \sum_{i=1}^N F(e_i) = \sum_{i=1}^N g(P (e_i)).
\]
It follows that $F$ is a 
       Gleason function of weight $W$ for the ONBs for $\KK^N$.

By Gleason's Theorem \ref{thm:saimpliesglea} (which applies to $F$ because $F$ is 
non-negative and 
$n \geq 3$), there exists a necessarily positive self-adjoint operator $B: \HH \to \HH$ 
such that $F(x) = \inner{B(x)}{x}$ for all $x \in S^{N-1}$.

Set $A := P B P$. We claim $g(x) = \inner{A(x)}{x}$ for all $x \in B^d$. Let $x \in B^d$. 
If $x = 0$ then the claim is $g(0) = 0$, which has been proven already. Otherwise,
 $x \neq 0$, and $y := x / \norm{x}$ is a unit-norm vector. Since $y \in \HH$, $y = P(y)$. 
 Using Theorem \ref{thm:nonnegativehomog} and the fact that $P$ is self-adjoint, we obtain
\[
g(x) = g(\norm{x} y) = \norm{x}^2 g(y) = \norm{x}^2 \inner{B(y)}{y} 
\]
\[
= 
\norm{x}^2 \inner{B(P(y))}{P(y)} = \norm{x}^2 \inner{A(y)}{y} = \inner{A(x)}{x}
\]
as claimed. Note that $A$ is a self-adjoint operator $\HH \to \HH$, since
\[
A^* = (PBP)^* = P^* B^* P^* = P B P = A.
\]
Thus, 
the spectral theorem gives an ONB $\{ u_i \}_{i=1}^d$ for $\HH$ consisting of eigenvectors 
of $A$; say $A (u_i)= \lambda_i u_i$ for each $i$. Then, $\{ u_i \}_{i=1}^d$ is a Parseval 
frame for $\HH$, so that
\[
W = \sum_{i=1}^d g(u_i) = \sum_{i=1}^d \inner{A (u_i)}{u_i} 
= \sum_{i=1}^d \lambda_i \norm{u_i}^2 = \sum_{i=1}^d \lambda_i = {\rm tr} (A),
\]
completing the proof of the claim.
\end{proof}

Theorem \ref{thm:nonnegativemainthm} extends to similar theorems about bounded 
Gleason functions, as we demonstrate.

\begin{thm} 
\label{thm:realmainthm}
Let $\HH = \KK^d$. A bounded, real-valued function $g: B^d \to \RR$ is a Gleason function of 
weight $W$ for the finite Parseval frames for $\HH$ if and only if there exists a self-adjoint operator 
$A: \HH \to \HH$ with trace ${\rm tr}(A) = W$ such that
\[
\forall x \in B^d, \quad g(x) = \inner{A(x)}{x}.
\]
\end{thm}

\begin{proof}
For the ``if'' direction see Theorem \ref{thm:gleapars}.

For the ``only if'' implication, let $\lambda := \sup_{x \in B^d} |g(x)|$. We first claim that 
$|g(x)| \leq \lambda \norm{x}^2$. Suppose by contradiction that this is not the case; then,
there exists $y \in B^d$ such that 
$|g(y)|> \lambda \norm{y}^2$. In particular, there exists some $\epsilon > 0$ such that 
$|g(y)| > (\lambda + \epsilon) \norm{y}^2$. 
By Proposition \ref{prop:prop2}, we have 
$g(0) = 0$,
and so $y \neq 0$. Since $\lambda / (\lambda + \epsilon) < 1$, 
there exists a positive rational number $a$ satisfying
\[
  \frac{\lambda}{(\lambda + \epsilon) \norm{y}^2} \leq a \leq \frac{1}{\norm{y}^2}.
\]
Set $z = \sqrt{a}y$. Then, $\norm{z}^2 = a \norm{y}^2 \leq 1$ so that
$z \in B^d$. Furthermore, by Lemma \ref{lem:rationalapprox},
\[
|g(z)| = a |g(y)| > a (\lambda + \epsilon) \norm{y}^2 \geq (\lambda + \epsilon) \frac{\lambda}{\lambda + \epsilon} = \lambda.
\]
Hence $|g(z)| > \lambda$, contradicting the choice of $\lambda$.

Now define an auxiliary function $f(x) = g(x) + \lambda \norm{x}^2$. Then, 
$f$ is non-negative, since for any $x \in B^d$, we have
\[
f(x) \geq \lambda \norm{x}^2 - |g(x)| \geq 0.
\]
Furthermore, $f$ is a Gleason function of weight $W + \lambda d$ for the finite Parseval 
frames for $\HH$. By Theorem \ref{thm:nonnegativemainthm}, there exists a positive 
semi-definite self-adjoint operator $B$ such that $f(x) = \inner{B(x)}{x}$ for all $x \in B^d$. 
Then, $g(x) = \inner{x}{(B - \lambda I)(x)}$, 
where $B - \lambda I$ is self-adjoint. The theorem follows by setting $A := B - \lambda I$, 
noting that ${\rm tr} (A) = {\rm tr} (B - \lambda I) = W + \lambda d - \lambda d = W$.
\end{proof}

It is now not difficult to extend this result to the complex case.

\begin{thm} 
\label{thm:complexmainthm}
Let $\HH = \CC^d$. A bounded function $g: \HH \to \CC$ is a Gleason function of weight 
$W$ for the finite Parseval frames for $\HH$ if and only if there exists a linear operator 
$A: \HH \to \HH$ with trace ${\rm tr}(A) = W$ such that
\[
\forall x \in B^d, \quad g(x) = \inner{A(x)}{x}.
\]
\end{thm}
\begin{proof}
For the ``if'' implication see Remark \ref{rem:normal}$a$.

For the ``only if'' implication, observe that the real and imaginary parts of $g$ are 
themselves Gleason functions for the finite Parseval frames for $\HH$. Thus 
$g = u + iv$ where $u$ and $v$ are both real-valued Gleason functions for the finite 
Parseval frames for $\HH$. Since $|g| = \sqrt{u^2 + v^2}$ is bounded, so too are 
$u$ and $v$, and Theorem \ref{thm:realmainthm} applies. Hence there exist self-adjoint 
operators $B$ and $C$ such that $u(x) = \inner{B(x)}{x}$ and $v(x) = \inner{C(x)}{x}$ 
for all $x \in B^d$.  Let $A = B + iC$; then 
$g(x) = \inner{A(x)}{x}$ for all $x \in B^d$. Since $W$ is equal to the weight of $u$ 
plus $i$ times the weight of $v$, we have ${\rm tr}(A) = {\rm tr}(B) + i {\rm tr}(C) = W$.
\end{proof}

%%%%%%%%%%%%%%%%%%%%%
%%%%%%%%%%%%%%%%%%%%%

\section{Gleason functions of degree $N$}
\label{sec:glean}

%%%%%%%%%%%%%%%%%%%%%

\subsection{Inclusion theorem and a problem}
\label{sec:thmprob}

 Let ${\mathcal P}_N$ be the set of Parseval frames for $\HH = {\mathbb K}^d$,
for which each $P \in {\mathcal P}_N$ has $N \geq d$ elements.

\begin{defn}[Gleason functions of degree $N$]
\label{defn:gleadegn}
Let $\HH = \KK^d$. A function $g:B^d \longrightarrow {\mathbb K},\,B^d \subseteq \HH$, 
is a {\it Gleason function of degree $N$ and weight $W = W_{g,N} \in \KK$}  
for the set $\Pcal_N$ of Parseval frames for $\HH$
if 
\[
\forall X=\{x_j\}_{j=1}^N \in {\mathcal P}_N, \quad \sum_{j=1}^N\,g(x_j) = W.
\]
Also, ${\mathcal G}_N$ designates the set of bounded Gleason functions of degree $N$ 
and any weight.
\end{defn}

The proof of the following result is the same as that of Theorem \ref{thm:gleapars}.

\begin{thm}
\label{thm:sagleapars}
 Let $\HH = \KK^d$, and let $A :\HH \to \HH$ be a self-adjoint operator.
 The function $g: B^d \to \KK$, defined by the formula,
 \begin{equation}
\label{eq:gleabd}
       \forall x\in B^{d}, \quad g(x)=\langle A(x),x \rangle,
\end{equation}
is a Gleason function of degree $N$, for any $N \geq {\rm dim} \HH$, and weight $W = {\rm tr} (A)$ 
for the set $\Pcal_N$ of Parseval frames for $\HH$.
\end{thm}

\begin{example}[Gleason functions for ONBs and Parseval frames]
\label{ex:gauss}
 Let $\HH = \KK^d$. {\it a}. Clearly, the function $g: B^d \to \RR$ defined by 
$g(x) = e^{\norm{x}^2} - 1$ is constant on $S^{d-1}$ and therefore is a bounded Gleason 
function of weight $W = d(e - 1)$ for the ONBs for $\HH$.

{\it b.}
We shall show that $g$ is not a Gleason function of degree $N>d$ for the Parseval 
frames for $\HH$. To this end, let $\{x_j\}_{j=1}^N$ be an equi-normed Parseval frame 
for $\HH$, see Proposition \ref{prop:pars} {\it d} as well as part {\it c} below; and let 
$\{y_j\}$ be an ONB for $\HH$ with $N-d$ copies of the zero vector adjoined, so that 
it too is a Parseval frame for $\HH$. Then, we have
\[
   \sum_{j=1}^N g(x_j)=Ne^{\frac{d}{N}} - N \neq de^1 + (N-d)e^0 - N  = \sum_{j=1}^N g(y_j),
\]
where the inequality is clear. The fact that the sums are unequal proves that $g$ is not a 
Gleason function of degree $N>d$ for the Parseval frames for $\HH$.

{\it c.} The fact that there are equi-normed Parseval frames for $\HH$ having $N > d$ 
elements fits into the theory of harmonic frames \cite{wald2018}, \cite{AndBenDon2019}, 
which itself is part of the group frame theory mentioned in Definition \ref{defn:frame}. 
For an explicit calculation to show the existence of equi-normed Parseval frames for 
$\CC^d$, consider the $N \times N$ DFT matrix, let $c > 0$, and let 
$s: \{ 1, \dots, d \} \to \{ 1, \dots, N \}$ be strictly increasing. Consider the $N$ vectors 
$\{ x_m \}_{m=1}^N$,
\[
    x_m = c (e^{2 \pi i ms(1)/N}, \ldots, e^{2 \pi i ms(d)/N})\in \CC^d.    
\]
For $z = (z_1, \ldots, z_d) \in \CC^d$, we compute
\[
 \sum_{m=1}^N \, | \langle z, x_m \rangle |^2
  = c^2 \sum_{j,k=1}^d z_j \overline{z_k} \big(\sum_{m=1}^N e^{2\pi i m(s(k) - s(j))/N} \big)
 \]
 \[
     = c^2 \sum_{j \neq k} z_j \overline{z_k} \big(\sum_{m=1}^N e^{2\pi i m(s(k) - s(j))/N} \big)
     +  c^2 \sum_{j=1}^d |z_j|^2 (\sum_{m=1}^N 1) = N c^2 \norm{z}^2.
\]
Thus, $\{ x_m \}_{m=1}^N$ is a Parseval frame when  $c = 1/{\sqrt{N}}$,
and in this case we compute that $\norm{x_m} = \sqrt{d/N}$ for each $m$.
In particular, each $x_m \in B^d$ as asserted in Proposition \ref{prop:pars} {\it e.}
\end{example}

Example \ref{ex:gauss} leads to the following problem.

\noindent

{\bf Problem}
\label{prob:inclusion}
 {\it a.} Let $\HH = \KK^d$ and $N > d$. Note that
$\mathcal{G}_N\subseteq \mathcal{G}_{N-1}$. To see this,
let $g\in \mathcal{G}_N$ have weight $W_{g,N}$, and let
$\{ y_j \}_{j=1}^{N-1} \in \mathcal{P}_{N-1}$. Then,
$\{ x_j \}_{j=1}^N \in  \mathcal{P}_{N}$, where $x_j = y_j$ for $1 \leq j \leq N-1$
and $x_N = 0 \in \HH$, since
\[
   \forall x \in \KK^d, \quad \norm{x}^2 = \sum_{j=1}^{N-1}\, |\inner{x}{y_j}|^2
   =  \sum_{j=1}^{N}\, |\inner{x}{x_j}|^2.
\]
Thus,
\[
        \sum_{j=1}^{N-1} g(y_j)=\sum_{j=1}^{N-1} g(y_j) + g(0) - g(0)
        =\sum_{j=1}^{N} g(x_j) -  g(0) = W_{g,N} -  g(0),
\]
i.e., $g\in \mathcal{G}_{N-1}$ with weight $W_{g,N-1} = W_{g,N} - g(0)$.

 {\it b.} We also have for $N>d$ that
$\mathcal{G}_{N}\subsetneq \mathcal{G}_d$ due to Example \ref{ex:gauss}.
Therefore,
\[
    \forall N > d,\quad {\mathcal G}_{N+1} \subseteq {\mathcal G}_{N} \subseteq \cdots
    \subsetneq {\mathcal G}_d.
\]
The {\it problem} is to resolve if the inclusions are proper when $N > d$.

We shall prove Theorem \ref{thm:GNsolution}, which, when combined with part {\it a} of the Problem,
allows us to assert that
\begin{equation}
\label{eq:GNequality}
       \forall N \geq d + 2,\quad {\mathcal G}_{N+1} = {\mathcal G}_{N}.
\end{equation}

% DEAL WITH G(0) = 0 BELOW, I.E., GET G(0) = 0 AND USE THM 5.4 (5) STATEMENT.

\begin{thm}
\label{thm:GNsolution}
Let $\HH = \KK^d$, and assume $N \geq d+ 2$. Then, every bounded Gleason function $g$ 
of degree $N$ and weight $W$ for the set $\Pcal_N$ of Parseval frames for $\HH$ is also 
a Gleason function of degree $N + 1$ and weight $W + g(0)$ for the set $\Pcal_{N+1}$ of 
Parseval frames for $\HH$.
\end{thm}

%\begin{thm}
%Let $\HH = \KK^d$, $d \neq 2$. Then, for $N \geq d + 2$,
%every non-negative Gleason function of degree $N$ and weight $W$ for the Parseval frames for $\HH$
%is also a Gleason function of weight $W$ for the Parseval frames for $\HH$.
%\end{thm}

The case $N <  d$ is not important because no Parseval frames with fewer than $d$ elements 
exist for $\KK^d$.
In the case $N = d$, there exist Gleason functions of degree $N$ but not of degree $N + 1$, 
as shown in Example \ref{ex:gauss}. For $N = d + 1$ and $d = 1$, it is known that 
$\mathcal G_{N+1} \subsetneq \mathcal G_N$, see Example \ref{ex:counterex2}. It is unknown whether $\mathcal G_{d+2} \subsetneq \mathcal G_{d+1}$ for $d > 1$.

The proof of Theorem \ref{thm:GNsolution} will be given in Subsection \ref{sec:GNsolution}.

%%%%%%%%%%%%%%%%

\subsection{Proof of Theorem \ref{thm:GNsolution}}
\label{sec:GNsolution}

As in the proof of Theorem \ref{thm:nonnegativemainthm}, we first consider the behavior 
of the functions of interest along lines through the origin. We use results from this case, 
combined with Gleason's theorem, to prove a weak version of Theorem \ref{thm:GNsolution}. 
From there we extend the result to the full Theorem \ref{thm:GNsolution}. 

Specifically, besides 
the theory developed in Section \ref{sec:gleapars}, the flow chart for proving Theorem
\ref{thm:GNsolution} is the following.
Lemma \ref{lem:1dim} is proved using Proposition \ref{prop:normsum}, and this lemma is used 
to prove Lemma \ref{lem:1dimhomogeneous}. Both lemmas have the setting $\HH = \KK^1$.
Lemma \ref{lem:1dimhomogeneous} is used in the proof of Theorem \ref{thm:GNsolutionGof0},
along with Theorem \ref{thm:Gleasonbounded}, which is an extension of Gleason's original theorem.
Theorem \ref{thm:Gleasonbounded} is the aforementioned weak version, and routine
adjustments allow us to obtain Theorem \ref{thm:GNsolution}.

\begin{prop}
\label{prop:normsum}
Let $\HH = \KK^1 = \KK$, and let $X = \{ x_i \}_{i=1}^N \subseteq \HH$, $N \geq 1$.  
Then, $X$ is a Parseval frame for $\HH$ if and only if
\begin{equation}
\label{eq:sum1}
    \sum_{i=1}^N \norm{x_i}^2 = 1.
\end{equation}
\end{prop}

\begin{proof}
{\it i.} Let $e_1$ be the standard basis vector for $\HH$ and so $\norm{e_1} = 1$, and let
$y, z \in \HH$. Then, $y = y' e_1$ and $z = z' e_1$ for some $y', z' \in \KK$. Thus,
$\norm{y} = |y'|$ and $\norm{z} = |z'|$, and, hence,
\begin{equation}
\label{eq:1diminnerproduct}
    |\inner{y}{z}|^2 = |\inner{y' e_1}{z' e_1}|^2 = |y'|^2 |z'|^2 |\inner{e_1}{e_1}| = |y'|^2 |z'|^2
    = \norm{y}^2 \norm{z}^2.
\end{equation}

{\it ii.} Assume \eqref{eq:sum1}. Using \eqref{eq:1diminnerproduct} we verify the
Parseval frame condition as follows. Let $x \in \HH$ and calculate
\begin{equation*}
    \sum_{i=1}^N |\inner{x}{x_i}|^2 = \sum_{i=1}^N \norm{x}^2 \norm{x_i}^2 =
    \norm{x}^2 \sum_{i=1}^N \norm{x_i}^2 = \norm{x}^2.
\end{equation*}

{\it iii.} Assume $X$ is a Parseval frame for $\HH$. Since $ |\inner{e_1}{x_j}|^2
= \norm{e_1}^2 \norm{x_j}^2 = \norm{x_j}^2$ by \eqref{eq:1diminnerproduct}, we have
\[
    \sum_{j=1}^N  |\inner{e_1}{x_j}|^2 =  \sum_{j=1}^N \norm{x_j}^2;
\]
but by the Parseval assumption on $X$ the left side is $\norm{e_1}^2 = 1$,
and this gives \eqref{eq:sum1}.
\end{proof}

Using Proposition \ref{prop:normsum}, we can establish the 1-dimensional special case 
(where $g(0) = 0$) of Theorem \ref{thm:GNsolution}. The proof of the 1-dimensional case 
presented below does not depend crucially on the hypothesis $g(0) = 0$, assuming, of course, 
that the necessary changes to the statement are made. We include the hypothesis $g(0) = 0$ 
merely because it will be important in a future step of the proof of Theorem \ref{thm:GNsolution}.

\begin{lem}
 \label{lem:1dim}
Let $\HH = \KK^{1} = \KK$, and assume $g$ is a Gleason function of degree $N \geq 3$ 
and weight $W$ 
for the set $\Pcal_N$ of Parseval frames for $\HH$. Furthermore, assume 
that $g(0) = 0$. Then, $g$ is a Gleason function of degree $N + 1$ and weight $W$ for the set 
$\Pcal_{N+1}$ of Parseval frames  for $\HH$.
\end{lem}

\begin{proof}
Let $X = \{ x_i \}_{i=1}^{N+1}$ be a Parseval frame for $\HH$, and let $e_1$ be the standard 
basis vector for $\HH$ and so $\norm{e_1} = 1$. Since $\{ x_i \}_{i=1}^{N+1}$ is a Parseval 
frame for $\HH$, each $x_i$ has $\norm{x_i} \leq 1$. Also, $1 - \norm{x_1}^2 - \norm{x_2}^2 \geq 0$ 
by Proposition \ref{prop:normsum}. Further, and also by Proposition \ref{prop:normsum}, 
we see that
\begin{equation*}
    \Big\{ x_1, x_2, \sqrt{1 - \norm{x_1}^2 - \norm{x_2}^2} \cdot e_1 \Big\}
\end{equation*}
is a 3-element Parseval frame for $\HH$.
Similarly,
\[
    \Big\{ \sqrt{\norm{x_1}^2 + \norm{x_2}^2} \cdot y, \sqrt{1 - \norm{x_1}^2 - \norm{x_2}^2} \cdot e_1, 0 \Big\}
\]
is a 3-element Parseval frame for $\HH$. By appending $N - 3$ copies of the $0$-vector
to these two sequences, we obtain $N$-element Parseval frames for $\HH$. Because $g$ is a 
Gleason function of degree $N$ and weight $W$ for the Parseval frames for $\HH$, we have
\[
    g(x_1) + g(x_2) + g\Big(\sqrt{1 - \norm{x_1}^2 - \norm{x_2}^2} \cdot e_1\Big) + (N - 3) g(0) = W
\]
and
\[
    g\Big(\sqrt{\norm{x_1}^2 + \norm{x_2}^2} \cdot e_1\Big) + g\Big(\sqrt{1 - \norm{x_1}^2 - \norm{x_2}^2} \cdot e_1\Big) + (N - 2) g(0) = W,
\]
so that
\begin{equation}
\label{eq:x1x2}
    g(x_1) + g(x_2) = g\Big(\sqrt{\norm{x_1}^2 + \norm{x_2}^2} \cdot e_1\Big).
\end{equation}

Also note that
\[
    \Big\{\sqrt{\norm{x_1}^2 + \norm{x_2}^2} \cdot e_1 \Big\} \cup \{ x_i \}_{i=3}^{N+1}
\]
is an $N$-element Parseval frame for $\HH$ by Proposition \ref{prop:normsum}, since
\[
    \norm{\sqrt{\norm{x_1}^2 + \norm{x_2}^2} \cdot e_1}^2 = \norm{x_1}^2 + \norm{x_2}^2.
\]
Thus, using \eqref{eq:x1x2}, we obtain
\begin{equation*}
    \sum_{i=1}^{N+1} g(x_i) = g\Big(\sqrt{\norm{x_1}^2 + \norm{x_2}^2} \cdot e_1\Big) + \sum_{i=3}^{N+1} g(x_i) = W,
\end{equation*}
where the second equality is a consequence of our assumption on $g$. Since $\{ x_i \}_{i=1}^{N+1}$ was an arbitrary $(N + 1)$-element Parseval frame for $\KK^1$, it follows that $g$ is a Gleason function of degree $N + 1$ and weight $W$ for the Parseval frames $\Pcal_{N+1}$ for $\HH$.
\end{proof}

\begin{example}[Lemma \ref{lem:1dim} and $N=2$]
\label{ex:counterex2}
Lemma \ref{lem:1dim} is false for $N = 2$ and counterexamples are not difficult to construct,
as we now illustrate. Let $\mathbb{H} = \KK^1$.

{\it a.}  Choose an arbitrary $\epsilon \in (0, 1/3)$. 
Define the function $g: B^1 \to \KK$ by
\[
g(x) = \begin{cases}
\norm{x}^2,  & \norm{x}^2 \not\in \{ \epsilon, 1 - \epsilon \}, \\
1 - \epsilon, & \norm{x}^2 = \epsilon \\
\epsilon, & \norm{x}^2 = 1 - \epsilon.
\end{cases}
\]
This $g$ is a Gleason function of degree $2$ and weight $1$ for the 
set $\Pcal_2$ of Parseval frames for
$\HH$. To see this, suppose $\{ x_1, x_2 \}$ is a Parseval frame for $\HH$.

If either of $\norm{x_1}^2$ or $\norm{x_2}^2$ is 
one of the elements of the set $\{ \epsilon, 1 - \epsilon \}$, then the other squared norm must be 
the other element of the set $\{ \epsilon, 1 - \epsilon \}$ by  Proposition \ref{prop:normsum}.
Thus, if $\norm{x_1}^2 = \epsilon$, then
$g(x_1) = 1 - \epsilon$ and $g(x_2) =\epsilon$, yielding $g(x_1) + g(x_2) = 1$, 
which is the desired Gleason function property, with a similar calculation when 
$\norm{x_1}^2 = 1 - \epsilon$.
 
If 
$\norm{x_1}^2 \not\in \{ \epsilon, 1 - \epsilon \}$, then 
$\norm{x_2}^2 \not\in \{ \epsilon, 1 - \epsilon \}$ by Proposition \ref{prop:normsum}. Consider the parabola 
$h(x) := \langle I(x), x \rangle$ defined on $B^d$,
noting that the trace of the identity mapping $I$ is $1$.
We have $h(x_1) + h(x_2) = 1$ by Theorem \ref{thm:gleapars}.
On the other hand, $g$ coincides with $\inner{x}{x}$ at $x_1$ and $x_2$, 
and hence $g(x_1) + g(x_2) = h(x_1) + h(x_2) = 1$, which again is the desired Gleason function property.

{\it b.} However, $g$ is not a Gleason function of degree $N > 2$ and any weight for the Parseval
frames for $\HH$. 

To see this, let $e_1$ be the standard basis vector for $\HH$, so that 
$\norm{e_1} = 1$. Given $N \geq 2$, one can construct an $N$-element Parseval 
frame $\{ x_i \}_{i=1}^N$ for $\HH$ by setting 
$x_1 = e_1$ and $x_2 = x_3 = \dots = x_N = 0$. This is a Parseval frame for $\HH$ by Proposition \ref{prop:normsum}. We have
\[
      \sum_{i=1}^N g(x_i) = g(x_1) + (N - 1)g(0) = 1.
\]
However, we can also construct an $N$-element Parseval frame $\{ y_i \}_{i=1}^N$ for
$\HH$ by setting $y_1 = y_2 = \sqrt{\epsilon} \cdot e_1$, $y_3 = \sqrt{1 - 2\epsilon} \cdot e_1$, 
and $y_4 = y_5 = \dots = y_N = 0$. This is a Parseval frame for $\HH$ by Proposition \ref{prop:normsum}. 
Observe that $1 - 2\epsilon \not\in \{ \epsilon, 1 - \epsilon\}$ since $\epsilon \in (0, 1/3)$. Hence,
\[
\sum_{i=1}^N g(y_i) = g(y_1) + g(y_2) + g(y_3) + (N - 3) g(0) = 2(1 - \epsilon) + (1 - 2\epsilon) + 0 = 3 - 4\epsilon.
\]
Since $3 - 4\epsilon > 5 / 3 > 1$, $g$ cannot be a Gleason function of degree $N > 2$ and any
weight for the Parseval frames for $\HH$.

\end{example}

The following lemma is elementary to prove given Lemma \ref{lem:1dim}, and it is 
crucial for the next step in the proof of Theorem \ref{thm:GNsolution}.

\begin{lem} 
\label{lem:1dimhomogeneous}
Let $\HH = \KK^1$, and let $g: B^1 \to \KK$ be a Gleason function of degree $N \geq 3$ 
and weight $W$ 
for the set $\Pcal_N$ of Parseval frames for $\HH$. Assume that $g$ is bounded 
and $g(0) = 0$. Then,
\[
\forall \alpha \in \KK, \; |\alpha| \leq 1, \; {\rm and } \; \forall x \in B^1,
\hspace{2em} g(\alpha x) = |\alpha|^2 g(x).
\]
\end{lem}
\begin{proof}
An induction based on Lemma \ref{lem:1dim} implies that $g$ is in fact a Gleason function 
of weight $W$ for \emph{all} finite Parseval frames for $\HH$. Then, 
Theorem \ref{thm:realmainthm} or 
Theorem \ref{thm:complexmainthm}, according to whether $\KK = \RR$ or $\CC$, respectively, 
implies that there is a linear operator $A: \HH \to \HH$ such that $g(x) = \inner{A(x)}{x}$ for 
all $x \in B^1$. Then, for $\alpha \in \KK$ with $|\alpha| < 1$,
\[
g(\alpha x) = \inner{A(\alpha x)}{\alpha x} = |\alpha|^2 \inner{A(x)}{x} = |\alpha|^2 g(x).
\]
\end{proof}

We now prove the special case of Theorem \ref{thm:GNsolution}, where $g(0) = 0$ 
and $g$ is bounded. 
The proof is similar to that of Theorem \ref{thm:nonnegativemainthm}.

\begin{thm} 
\label{thm:GNsolutionGof0}
Let $\HH = \KK^d$, and assume $N \geq d + 2$. Let $g$ be a Gleason function of degree $N$ 
and weight $W$ for the set $\Pcal_N$ of Parseval frames for $\HH$. Assume that $g$ is bounded 
and that $g(0) = 0$. Then, $g$ is a Gleason function of weight $W$ for all the  
Parseval frames for $\HH$.
\end{thm}

\begin{proof}
{\it i.}
Let $\HH_1$ be a one-dimensional subspace of $\HH$. Then, the restriction of $g$ 
	to $B^d \cap \HH_1$ is a bounded Gleason function of degree $N - d + 1 \geq 3$ and some 
	weight $W_1$ for the set $\Pcal_{N-d+1}$ of Parseval frames for $\HH_1$.
	
	To see this, let $\{ y_i \}_{i=1}^{d-1}$ be an 
	ONB for $(\HH_1)^\bot$. If $\{ x_j \}_{j=1}^{N - d + 1}$ is a Parseval frame for $\HH_1$, then 
	$\{ x_j \}_{j=1}^{N - d + 1} \cup \{ y_i \}_{i=1}^{d-1}$ is a Parseval frame 
	for $\HH$. 
	In fact, letting $P_1$ denote the orthogonal projection onto $\HH_1$ and $P_2$ 
	denote the orthogonal projection onto $(\HH_1)^\bot$, we have
	\begin{align*}
	\norm{x}^2 = \inner{x}{x} &= \inner{P_1 (x )+ P_2 (x)}{P_1 (x) + P_2 (x)} \\
	&= \norm{P_1 (x)}^2 + \inner{P_1 (x)}{P_2 (x)} + \inner{P_2 (x)}{P_1 (x)} + \norm{P_2(x)}^2 \\
	&= \norm{P_1 (x)}^2 + \norm{P_2 (x)}^2.
	\end{align*}
Thus, we compute
	\begin{align*}
	\norm{x}^2 = \norm{P_1 (x)}^2 + \norm{P_2 (x)}^2 &= \sum_{j=1}^{N - d + 1} |\inner{P_1 (x)}{x_j}|^2 
	+ \sum_{i=1}^{d-1} |\inner{P_2 (x)}{y_i}|^2 \\
	&= \sum_{j=1}^{N - d + 1} |\inner{x}{x_j}|^2 + \sum_{i=1}^{d-1} |\inner{x}{y_i}|^2,
	\end{align*}
	because $\{ x_j \}_{j=1}^{N - d + 1}$ is a Parseval frame for $\HH_1$, $\{ y_i \}_{i=1}^{d-1}$ 
	is an ONB for $(\HH_1)^\bot$, and $P_1$ and $P_2$ are self-adjoint. 
	By our assumption on $g$, we obtain
	\begin{equation*}
	\sum_{j=1}^{N - d + 1} g(x_j) + \sum_{i=1}^{d - 1} g(y_i) = W.
	\end{equation*}
	Therefore, $\sum_j g(x_j) = W - \sum_i g(y_i)$. Setting $W_1 := W - \sum_i g(y_i)$, the claim follows.

{\it ii.} Now let $x \in S^{d-1}$. Take $\HH_1 = \langle x \rangle$ in part {\it i.} Then, 
Lemma \ref{lem:1dimhomogeneous} gives 
\begin{equation}
\label{eqn:homog}
g(\alpha x) = |\alpha|^2 g(x)
\end{equation}
for any $\alpha \in \KK$ with $|\alpha| \leq 1$.

{\it iii.} Next, view $\HH$ as the subspace of the larger space $\KK^N$ spanned by the first $d$ 
standard basis vectors. Let $P: \KK^N \to \HH$ be the projection onto the first $d$ coordinates, 
and let $F(x) = g(P(x))$ for all $x \in S^{N-1}$. If $\{ e_i \}_{i=1}^N$ is any ONB for $\KK^N$, then
Proposition \ref{prop:frameprojection} implies that $\{ P (e_i )\}_{i=1}^N$ is a Parseval frame for $\HH$. 
Thus, we have
\[
\sum_{i=1}^N F(e_i) = \sum_{i=1}^N g(P (e_i)) = W.
\]
Since $\{ e_i \}_{i=1}^N$ was an arbitrary ONB for $\KK^N$, $F$ is a bounded
Gleason function for the ONBs for 
$\KK^N$. If $d = 0$ then the theorem holds trivially, so we may assume $d > 0$. Thus, 
$N = d + 2 \geq 3$, so that Theorem \ref{thm:Gleasonbounded} gives a linear operator 
$A: \HH \to \HH$ such that $F(x) = \inner{A(x)}{x}$ for all $x \in S^{N-1}$. In particular, 
for $x \in S^{d-1}$, we have
\[
g(x) = g(P(x)) = F(x) = \inner{A(x)}{x}.
\]
For any $y \in B^d$, we have either $y = 0$ (in which case $g(y) = 0 = \inner{A(y)}{y}$) or 
$0 < \norm{y} \leq 1$. In the latter case, we have
\[
g(y) = g\Big(\norm{y} \cdot \frac{y}{\norm{y}}\Big) = \norm{y}^2 g\Big(\frac{y}{\norm{y}}\Big) 
= \norm{y}^2 \Big\langle A \Big(\frac{y}{\norm{y}}\Big), \frac{y}{\norm{y}} \Big\rangle = \inner{A(y)}{y}.
\]
By Theorem \ref{thm:gleapars} and Remark \ref{rem:normal} {\it b}, we see that $g$ is a 
Gleason function for all the Parseval frames for $\HH$ (not just the finite ones).
\end{proof}

Once these facts have been established, the proof of Theorem \ref{thm:GNsolution} is straightforward.

\begin{proof}[Proof of Theorem \ref{thm:GNsolution}.]
Observe that $f: B^d \to \KK$, defined by $f(x) := g(x) - g(0)$, is a bounded Gleason function of 
degree $N$ and weight $W - N g(0)$ for the set $\Pcal_N$ of Parseval frames for $\HH$. Indeed, 
if $\{ x_j \}_{j=1}^N$ is a Parseval frame for $\HH$, then
\[
\sum_{j=1}^N f(x_j) = \Big(\sum_{j=1}^N g(x_j)\Big) - N g(0) = W - N g(0).
\]
Furthermore, $f(0) = 0$. Hence Theorem \ref{thm:GNsolutionGof0} implies that $f$ is a Gleason 
function of weight $W - N g(0)$ for the finite Parseval frames for $\HH$. In particular, for any 
$(N + 1)$-element Parseval frame $\{ y_i \}_{i=1}^{N+1}$ for $\HH$, we have
\[
\sum_{i=1}^{N+1} g(y_i) = \Big(\sum_{i=1}^{N+1} f(y_i)\Big) + (N + 1) g(0) 
= W - N g(0) + (N + 1) g(0) = W + g(0).
\]
Thus, $g$ is a Gleason function of degree $N + 1$ and weight $W + g(0)$ for the set 
$\Pcal_{N+1}$ of Parseval frames for $\HH$.
\end{proof}

%%%%%%%%%%%%%%%%%%%
%%%%%%%%%%%%%%%%%%%%%%%%%%%%%%%

\section{An application of Gleason functions}
\label{sec:appl}

Theorem \ref{thm:GNsolution} has an application in quantum measurement
with regard to the theory developed by Busch in \cite{busc2003}. To see this, let us begin
with Definition \ref{def:genprob} taken from \cite{busc2003}.
The set $\mathcal{E}(\HH)$ of operators on $\HH$ was defined in Definition \ref{defn:povmsaop}.

\begin{defn} 
	\label{def:genprob}
	A {\it generalized probability measure} on 
	$\mathcal{E}(\HH)$ is a function, $v\colon \mathcal{E}(\HH) \to \RR$, with the following properties:
	\begin{enumerate}
		\item $0 \leq v(E) \leq 1$ for all $E \in \mathcal{E}(\HH)$;
		\item $v(I) = 1$;
		\item $v(\sum_{j \in J} E_j) = \sum_{j \in J} v(E_j)$ for all countable indexed families $\{ E_j \}_{j \in J} \subset \mathcal{E}(\HH)$ for which
		$\sum_{j \in J} E_j \in \mathcal{E}(\HH)$.
	\end{enumerate}
\end{defn}

Busch characterized the generalized probability measures on Hilbert spaces of the type 
encountered in quantum mechanics as follows. 

\begin{thm}[Busch]
	\label{thm:busch}
	Let $\HH$ be a separable complex Hilbert space, and let $v$ be a generalized probability measure on $\mathcal{E}(\HH)$. Then, there exists a density 
	operator $\rho$ on $\HH$ such that $v(E) = {\rm tr}(\rho E)$ for all $E \in \mathcal{E}(\HH)$.
    (Recall that a density operator is a positive semi-definite trace class operator with trace $1$.)
\end{thm}

\begin{rem} {\it a.}
	From the perspective of quantum mechanics, elements of $\mathcal{E}(\HH)$\
	can be interpreted 
	as physical {\it effects}, while generalized probability measures on $\mathcal{E}(\HH)$ 
	can be interpreted as physical {\it states}. 
	Hence, Theorem \ref{thm:busch} asserts that states can be represented by density 
	operators in a similar fashion to the result of Theorem \ref{thm:glea2}.
	Some comparison of these theorems is in order. Gleason's Theorem \ref{thm:glea2} is 
	concerned with measures on the closed subspaces of a Hilbert space, whereas Busch's 
	Theorem \ref{thm:busch} is concerned with measures on the effects of a Hilbert space, cf. 
	Caves et al. \cite{CavFucManRen2004} (2004).
	Both admit similar physical interpretations. 
	Busch's theorem is valid when the Hilbert space $\HH$ has dimension 2.
	
{\it b.} Busch's theorem is striking, useful, and weaker than Gleason's theorem. It
is essentially Gleason' theorem for POVMs; and 
it is weaker since $v$ is defined on a much larger
space of operators than in Gleason's setting.

\end{rem}

We shall use Theorem \ref{thm:GNsolution} to prove that
if $\HH = \KK^d, \, d \geq 2$, then
condition (3) in Definition \ref{def:genprob} can be replaced 
by the seemingly weaker condition,

(3$'$) There exists $N \geq \dim \HH + 2$ such that $\sum_{i=1}^N v(E_i) = 1$ 
whenever $\{E_i\}_{i=1}^N$ is an $N$-element POVM on $\HH$.

This is made precise by the following theorem:

\begin{thm}
\label{thm:genprob}
Let $\HH = \KK^d, \, d \geq 2$. Suppose
$v\colon \mathcal{E}(\HH) \to \RR$ is a non-negative function 
for which $v(I) = 1$. Furthermore, assume that there exists $N \geq d + 2$ 
such that $\sum_{i=1}^N v(E_i) = 1$ whenever $\{E_i\}_{i=1}^N$ is an $N$-element POVM 
on $\HH$. Then, $v$ is a generalized probability measure on $\mathcal{E}(\HH)$.
\end{thm}

\begin{proof}
{\it i.} Define a function $g_v$ on the closed unit ball $B^d$ of $\HH$ by $g_v(x) 
= v(x \otimes  x^*)$, recalling that the 
 tensor product $x \otimes  x^* : \HH \times \HH' \to \Lcal(\HH)$ is the outer product
$xx^*$. We claim that $g_v$ 
is a Gleason function of degree $N$ and weight 1 for {\it all} of the Parseval frames for $\HH$. 

Clearly, $g_v$ is 
non-negative by its definition and the hypothesis on $v$.
If $\{x_i\}_{i=1}^N$ is an $N$-element Parseval frame for $\HH$, 
then $\{ x_i \otimes x_i^*\}_{i=1}^N$ is a POVM on $\HH$ by Proposition \ref{prop:povm}. 
Therefore,
by our POVM assumption, we have
\begin{equation*}
\sum_{i=1}^N g_v(x_i) = \sum_{i=1}^N v(x_i \otimes x_i^*) = 1,
\end{equation*}
and so $g_v$ is a Gleason function of degree $N$ and weight $1$ for the set $\Pcal_N$ of
Parseval frames for $\HH$.

Also, note that $\{I, 0, \dots, 0\}$, where there are $N - 1$ copies of the $0$-operator, 
is a POVM on $\HH$, 
and so, by our POVM assumption again, we have  
\begin{equation*}
      v(I) + (N - 1)v(0) = 1 + (N - 1)v(0) = 1.
\end{equation*}
Thus, since $N \geq d + 2 > 1$, we obtain $v(0) = 0$ for the $0$-operator in the domain of $v$.
As a result we see that $g_v(0) = 0$ for $0 \in B^d$ in the domain of $g_v$.

From Theorem \ref{thm:GNsolution}, it follows that $g_v$ is a Gleason function of 
degree $N + 1$ and weight $1$ for the set $\Pcal_{N+1}$ of Parseval frames for $\HH$. 
A straightforward 
induction argument shows that $g_v$ is therefore a Gleason function of weight $1$ 
for all finite Parseval frames for $\HH$. Theorem \ref{thm:nonnegativemainthm}, or 
Theorems \ref{thm:realmainthm} or \ref{thm:complexmainthm}, apply to 
prove that $g_v$ is a quadratic form on $B^d$. From there, Theorem \ref{thm:gleapars} and
Remark \ref{rem:normal} {\it b} imply that $g_v$ is a Gleason function for 
all the Parseval frames for $\HH$ 
(not just the finite frames)
as claimed. We shall use this result in part {\it iii}.

{\it ii.} We shall show that if $E = \sum_{i=1}^d x_i \otimes x_i^*$, then $v(E) = \sum_{i=1}^d g_v(x_i)$. 
For this, note that both $\{ E, I - E \}$ and $\{ x_1 \otimes x_1^*, \dots, x_d \otimes x_d^*, I - E \}$ are
POVMs on $\HH$. Appending copies of $0$ to these POVMs until both have $N$ elements 
and applying the hypothesized condition on $v$, we obtain the equation
\begin{equation*}
    v(E) + v(I - E) + (N - 2)v(0) = 1 = \left(\sum_{i=1}^d v(x_i \otimes x_i^*)\right) + v(I - E) + (N - d - 1)v(0).
\end{equation*}
Using $v(0) = 0$ and canceling the $v(I - E)$ term shows that
\begin{equation*}
     v(E) = \sum_{i=1}^d v(x_i \otimes x_i^*) = \sum_{i=1}^d g_v(x_i),
\end{equation*}
as asserted.

{\it iii.} Now let $\{E_j\}_{j \in J} \subseteq \Ecal(\HH)$ be a countable sequence such that 
$\sum_{j \in J} E_j \in \mathcal{E}(\HH)$. Invoke the spectral theorem 
as in Proposition \ref{prop:pars}  to write $E_j = \sum_{i=1}^d x_{ij} \otimes x_{ij}^*$ 
for each $j \in J$, $\sum_{j \in J} E_j = \sum_{i=1}^d y_i \otimes y_i^*$, and 
$I - \sum_{j \in J} E_j = \sum_{i=1}^d z_i \otimes z_i^*$ for some collections of vectors 
$x_{ij}, y_i, z_i \in \HH$. Hence, we have the two equations,
\begin{align*}
       \sum_{i=1}^d z_i \otimes z_i^* + \sum_{i=1}^d y_i \otimes y_i^* &= \left(I - \sum_{j \in J} E_j\right) 
      + \sum_{j \in J} E_j = I,\\
     \sum_{i=1}^d z_i \otimes z_i^* + \sum_{j \in J} \sum_{i=1}^d x_{ij} \otimes x_{ij}^* &
     = \left(I - \sum_{j \in J} E_j\right) 
       + \sum_{j \in J} E_j = I.
\end{align*}
Since $(u \otimes u^*)(x) = \langle x,u \rangle u$ (Definition \ref{defn:tensorket}),  we can apply these 
operators on the left side of both equations to any $x \in \HH$, and then take the inner product
with $x$, to assert
that $\{ y_i \}_{i=1}^d \cup \{ z_i \}_{i=1}^d$ and $\{x_{ij}\}_{j \in J, i = 1, \dots, d} \cup \{ z_i \}_{i=1}^d$ 
are both Parseval frames for $\HH$. For example,
\[
 \norm{x}^2 = |\langle I(x),x \rangle| 
 =\big| \langle \big( \sum_{i=1}^d  z_i \otimes z_i^* \big)x + 
 \big(\sum_{i=1}^d y_i \otimes y_i^*\big)x, x \rangle \big|
\]
\[
 =\big| \langle  \sum_{i=1}^d   \langle x,z_i \rangle z_i + \sum_{i=1}^d \langle x,y_i \rangle y_i   ,x  \rangle \big|
 =\big|  \sum_{i=1}^d | \langle x,z_i \rangle|^2  + \sum_{i=1}^d  |\langle x,y_i \rangle|^2 \big|
 = \sum_{i=1}^d | \langle x,z_i \rangle|^2  + \sum_{i=1}^d  |\langle x,y_i \rangle|^2.
\]

Thus,
\begin{equation*}
\sum_{i=1}^d g_v(z_i) + \sum_{i=1}^d g_v(y_i) = 1 = \sum_{i=1}^d g_v(z_i) 
+ \sum_{j \in J} \sum_{i=1}^d g_v(x_{ij}),
\end{equation*}
so that
\begin{equation*}
    \sum_{i=1}^d g_v(y_i) = \sum_{j \in J} \sum_{i=1}^d g_v(x_{ij}).
\end{equation*}
Consequently, we obtain
\begin{equation}
\label{eq:additivity}
v\left( \sum_{j \in J} E_j \right) = \sum_{i=1}^d g_v(y_i) = 
      \sum_{j \in J} \sum_{i=1}^d g_v(x_{ij}) = \sum_{j \in J} v(E_j),
\end{equation}
where the last equality follows from part {\it ii.}
Equation \eqref{eq:additivity} is the desired countable additivity condition of Definition \ref{def:genprob}.
\end{proof}

%%%%%%%%%%%%%%%%%%%%%%%
%%%%%%%%%%%%%%%%%%%%%%
\appendix\section{}
\label{sec:app}

{\it a.} If $N > d^2$, then there is no ETF for $\CC^d$ consisting of $N$ elements;
and these values of $N$ can be viewed as a natural regime for the Grassmannian frames defined in
part {\it c.}
Further, if $N < d^2$, then there are known cases for which there are no ETFs, e.g., $d=3, \, N=8$
\cite{szol2017}. 
Determining compatible values of $d,\,N$ for which there are ETFs is a subtle,
unresolved, and
highly motivated problem, see, e.g., \cite{FicMixTre2012, FicMix2016, 
FicJasMixPet2018, wald2018}.

{\it b.} (ETFs and the Welch bound) The {\it coherence} or maximum correlation $\mu(X)$ of a set
$X = \{ x_j \}_{j=1}^N \subseteq \CC^d$ of unit norm elements is defined as
\begin{equation}
\label{eq:welch}
   \mu(X) = {\rm max}_{j \neq k}\, |\langle x_j, x_k \rangle|.
\end{equation}
Welch (1974) \cite{welc1974} proved the fundamental inequality,
\begin{equation}
\label{eq:welchineq}
     \mu(X) \geq \sqrt{\frac{N - d}{d(N - 1)}},
\end{equation}
that itself is important in understanding the behavior of the narrow band ambiguity function, 
see  
\cite{HerStr2009, BenBenWoo2012} and part {\it e}. The right side of 
the inequality (\ref{eq:welchineq}) is the {\it Welch bound}, 
cf. Proposition \ref{prop:parsprops} {\it f}. 
In the case that $X$ is a FUNTF for $\CC^d$,
then equality holds in (\ref{eq:welchineq}) if and only if $X$ is an ETF with constant 
$\alpha = \sqrt{\frac{N - d}{d(N - 1)}}$,
see \cite{HeaStr2003}, Theorem 2.3, as well as \cite{BenKol2006}, Theorem IV.2 
(Theorem 3) for
a modest but useful generalization. Because of the importance of Gabor
frames in this topic, we note that if $N = d^2$, then $\alpha = \sqrt{\frac{1}{d+1}}$.

{\it c.} (Grassmannian frames) If an ETF does not exist for a given $N \geq d+2$, then a
reasonable substitute is to consider $(N,d)$-Grassmannian frames.
Let $X = \{ x_j \}_{j=1}^N \subseteq \CC^d$ be a set of unit norm elements. $X$ is
an {\it $(N,d)$-Grassmannian frame} for $\CC^d$ if it is a FUNTF and if 
\[
       \mu(X) =  {\rm inf}\, \mu(Y),
\]
where the infimum is taken over all 
FUNTFs $Y$ for $\CC^d$ consisting of $N$ elements.
A compactness argument shows that $(N,d)$-Grassmannian frames exist,
see \cite{BenKol2006}, Appendix. Also,
ETFs are a subclass of Grassmannian frames, see
\cite{BodPauTom2009, wald2003}.
Further, as noted in \cite{HeaStr2003}, Grassmannian frames
have significant applicability, including spherical codes and designs,
packet based communication systems such as the internet,
and geometrically uniform codes
in information theory, and these last are essentially group frames \cite{forn1991} (1991), cf.
Definition \ref{defn:frame} {\it c} and \cite{BolEld2003}.

One of the major mathematical challenges is to construct Grassmannian frames,
 see \cite{BenKol2006, wald2018}.

{\it d.} (Zauner's conjecture) {\it Zauner's conjecture} is that for any
dimension $d \geq 1$ there is a FUNTF $X = \{ x_j  : j = 1,\ldots, d^2\}$ for $\CC^d$ 
such that
\[
     \forall \, j \neq k, \quad |\langle x_j, x_k \rangle | = \sqrt{\frac{1}{d + 1}}.
\]
The problem can be restated by asking if for each $d \geq 1$ there are $(d^2,d)$-Grassmannian frames
that
achieve equality with the Welch bound.
This is an open problem in quantum information theory, and
the conjecture by Zauner \cite{zaun1999} was motivated by issues
dealing with quantum measurement, cf. \cite{RenBluScoCav2004}.
There are solutions for some values of $d$, and solutions are referred to
as {\it symmetric, informationally complete, positive operator valued measures} (SIC-POVMs).
POVMs were introduced in Subsection \ref{sec:povm}. They
not only arise in quantum measurement and detection, e.g., see \cite{BenKeb2008},
Definition A.1,
but also draw on issues dealing with coherent states \cite{AliAntGaz2000}.
A major recent contribution to Zauner's conjecture is \cite{AppChiFlaWal2018}.

Zauner's conjecture is also related to frame potential energy in the following way.
In \cite{BenFic2003} FUNTFs were characterized as the minimizers of the $\ell^2$-
frame potential energy functional motivated by Coulomb's law. The $\ell^p$-version,
merely defined in \cite{BenKol2006}, was developed by
Ehler and Okoudjou, see \cite{EhlOko2012, EhlOko2013}.
 The main theorem in \cite{BenFic2003} proves the existence
of so-called Welch bound equality (WBE) sequences used for code-division 
multiple-access (CDMA) systems in communications, see \cite{MasMit1993, wald2003}. 
In fact, the essential inequality asserted
in the WBE setting of 
Massey and Mittelholzer \cite{MasMit1993} is an $\ell^2$-version of 
the $\ell^\infty$ inequality (\ref{eq:welchineq});
and the relevant equations in \cite{MasMit1993} are (3.4) -- (3.6). With this backdrop,
there is a compelling case relating solutions of Zauner's conjecture, as well
as Grassmannians, in terms of minimizers of all $\ell^p$-frame potentials,
see \cite{okou2016}.

{\it e.} (CAZAC sequences) Given a function $u : \ZZ/d\ZZ \longrightarrow \CC$. 
For any such $u$ we can define a Gabor FUNTF $U = \{u_j: j = 1, \ldots, d^2\}$,
where each $u_j $ consists of translates and modulations of $u$,
e.g., see \cite{pfan2013a}.

The discrete periodic ambiguity function $A(u)$ of $u$ is defined by the formula,
\[
       \forall \,(m, n) \in \ZZ/d\ZZ  \times \ZZ/d\ZZ, \quad A(u)(m,n)= 
      \frac{1}{d}\, \sum_{k=0}^{d-1}\,u(m+k)\,\overline{u(k)}\,e^{-2 \pi i k n/d}.
\]
 The function $u$
is a {\it constant
amplitude $0$-autocorrelation} (CAZAC) sequence if
\[
     \forall \,m\in\ZZ/d\ZZ,\quad |u(m)|=1, \quad \text{(CA)}
\]
and
\[
     \forall \,m\in\ZZ/d\ZZ\setminus\{0\},\quad \frac{1}{d} \sum_{k=0}^{d-1} \,u(m+k)\,\overline{u(k)}=0. \quad \text{ (ZAC)}.
\]
A recent survey on the theory and applicability of CAZAC sequences is \cite{BenCorMag2019}.
The construction of all CAZAC sequences remains a tantalizing and applicable venture.

A fundamental fact is the following theorem \cite{BenBenWoo2012}, Theorem 3.8.
{\it 
	Let $d = p$ be prime. There are explicit CAZAC sequences 
	$u : \ZZ/p\ZZ \longrightarrow \CC$ (due to Bj{\"o}rck)
	with the property that
	if $(m,n)\in(\ZZ/p\ZZ\times\ZZ/p\ZZ)\backslash\{(0,0)\}$, then
\[
	|A(u)(m,n)|\leq\frac{2}{\sqrt{p}}+\left\{\begin{array}{ll}
	\frac{4}{p}& \text{if } p\equiv1\, (\text{mod } 4)\\
	\frac{4}{p^{3/2}} & \text{if } p\equiv3\, (\text{mod } 4).
	\end{array}\right.
\]
In particular, $|A(u)(m,n)|\leq 3/\sqrt{p}$.} 

This implies that the coherence 
	$\mu (U)$ of $U$ satisfies the inequalities, 
\begin{equation}
\label{eq:bjorckcohineq}
     \frac{1}{\sqrt{p+1}} \leq \mu(U) \leq  \frac{3}{\sqrt{p}},
\end{equation}	
even though $|A(u)(m,n)|$ can have significantly smaller values than $3/{\sqrt{p}}$
for various $(m,n)$. This latter property hints at the deeper applicability 
of CAZAC sequences such as the Bj{\"o}rck sequence.

Because of the $0$-autocorrelation property, CAZAC sequences are the opposite of what
candidates for Zauner's conjecture should be. On the other hand, the inequality 
(\ref{eq:bjorckcohineq}) gives perspective with regard to Zauner's conjecture.
Further, these CAZAC sequences are an essential component of the background and
goals dealing with phase-coded waveforms that were the
driving force leading to the role 
of group frames in the vector-valued theory of \cite{AndBenDon2019}.

\begin{rem}[Uncertainty principle]
\label{rem:up}
{\it a.} It is relevant to understand {\it weighted extensions}
of Heisenberg's uncertainty principle in the context of 
a Gleason theorem for  Parseval frames, just as Gleason's original theorem
in the context of ONBs was driven by the Birkhoff and von Neumann
remark in Subsection \ref{sec:back}. These extensions 
of Heisenberg's uncertainty principle are both physically motivated
and use many techniques from harmonic analysis, see, e.g.,
\cite{BenHei1992} (1992), \cite{BenHei2003} (2003),
\cite{BenDel2017} (2017).

{\it b.} Because of the role of the uncertainty principle in quantum mechanics
and the technical role of graph theory in Schr{\"o}dinger eigenmap methods for
non-linear dimension reduction techniques, it is natural to continue the
development of graph theoretic uncertainty 
principles \cite{chun1997}, \cite{grun2003}, \cite{LamMae2011},
\cite{HamVanGri2011},
\cite{AgaLu-2013}, \cite{ShuNarFroOrtVan2013}, \cite{BenKop2015}, 
\cite{kopr2016},
\cite{ShuRicVan2015}, \cite{TsiBarDil2015a}, \cite{TsiBarDil2015b}.

\end{rem}

\nocite{*}
\bibliography{new_bib_01_06.bib}
\bibliographystyle{amsplain}

\end{document}